\documentclass[a4paper,abstract=true]{scrartcl}
\usepackage[margin=1in]{geometry}
\usepackage{amsfonts}
\usepackage{amssymb}
\usepackage{amsmath}
\usepackage{amsthm}
\usepackage{comment}
\usepackage{graphicx}
\usepackage{subcaption}
\usepackage{thm-restate}
\usepackage{mathrsfs}
\usepackage{mathalpha}

\usepackage[ruled,vlined,linesnumbered]{algorithm2e}
\usepackage{hyperref}
\usepackage[noabbrev,capitalize]{cleveref}

\usepackage{csquotes}
\usepackage{lineno}
\usepackage{complexity}
\usepackage{footnote}

\usepackage{tikz}
\usetikzlibrary{calc}
\usetikzlibrary{patterns,decorations.pathreplacing}

\usetikzlibrary{arrows.meta}

\newtheorem{theorem}{Theorem}
\newtheorem{observation}[theorem]{Observation}

\newtheorem{lemma}[theorem]{Lemma}

\newtheorem{cor}[theorem]{Corollary}

\newtheorem{proposition}[theorem]{Proposition}

\newtheorem*{definition*}{Definition}
\newtheorem*{lemma*}{Lemma}

\usepackage{authblk}

\allowdisplaybreaks

\begin{document}

	\title{Flip Distance Between Triangulations of Convex Polygons is NP-Complete}
	
	\author{Joseph Dorfer}
	
	\affil{Graz University of Technology, Graz, Austria, \textit{joseph.dorfer@tugraz.at}}

	\date{}
	
	\maketitle
	
	\begin{abstract}		
		
		The complexity of determining the minimum number of flips that transform one triangulation of a convex polygon into another has been raised as an open problem by Culik and Wood [Inf. Proc. Letters 1982] and has been popularized through the study of extremal pairs of triangulations by Sleator, Tarjan, and Thurston [STOC 1986\& J. Am. Math. Society~1988].
		
		The search for a hardness proof for the flip distance problem has yielded many (weaker) hardness-results in more general settings. Lubiw and Pathak [CCCG 2012, Comp. Geom. 2014] and Pilz [Comp. Geom. 2014] proved that computing the flip distance between triangulations of point sets in general position is \NP-complete. Further, Aichholzer, Mulzer, and Pilz [ESA~2013, DCG 2015] proved that computing the flip distance between triangulations of simple polygons is \NP-complete.
		
		The methods used to obtain these previous hardness results rely heavily on the geometric positioning of the vertices of the point set or the vertices of the polygon and are thus not feasible to prove \NP-hardness for the flip distance problem of triangulations of convex polygons, where vertices are placed on a circle.
		
		We formulate a notion of conflict graphs for flip sequences between triangulations of convex polygons that allows us to use both geometric and combinatorial tools to study flip sequences. Large acyclic subsets of the conflict graph will correspond to short flip sequences. We show \NP-completeness of the flip distance problem by proving that finding the largest acyclic subset of our conflict graphs is \NP-complete.
		
		Triangulations of convex polygons are counted by the famous Catalan numbers and are in bijection with many other famous structures. Therefore, they are an important cornerstone in the study of polytopes (in form of the associahedron) or lattices (in form of the Tamari lattice). As a consequence of our result, we obtain two results that are of independent interest in the area of data structures and discrete optimization: (1) The rotation distance problem of binary trees is \NP-complete; (2) The computation of shortest paths on the edge graph of a simple polytope is \NP-hard.
	\end{abstract}
	
	\newpage
	
	\section{Introduction}
	
	There are many combinatorial objects that can be counted by the well-known Catalan numbers $C_n =~\frac{1}{n+1}\binom{2n}{n}$~\cite{gardner2020time}. Many of these structures allow natural reconfiguration operations via small reconfiguration steps among them. There is a large family of structures for which this natural reconfiguration operation commutes with bijections between structures, that is, two structures that can be transformed into one another via a reconfiguration step can also be transformed into one another after applying a bijection that maps them to a different Catalan structure. Among them are flips in triangulations of convex polygons, rotations in binary trees and rearrangements of bracketings. In the 1960's it was discovered by Tamari and Stasheff that for this big family of structures the reconfiguration graph has a nice lattice structure, the Tamari Lattice~\cite{tamari1962algebra}, and arises as the 1-skeleton of a polytope, the so called associahedron~\cite{stasheff1963homotopy}.
	
	In 1982, Culik and Wood~\cite{CULIK198239} explicitly asked the question about the algorithmic complexity of determining the minimum number of rotations that are needed to transform one binary tree into one another. This question is equivalent to determining the minimum number of reconfiguration steps between any of the numerous structures that are in bijection with binary trees as long as the bijection preserves incidence under the reconfiguration operation. We will study this reconfiguration distance from the viewpoint of flips in triangulations of convex polygons.
	
	\subparagraph{Flips in triangulations of convex polygons.} Let $P$ be a convex polygon with $n$ vertices. Two edges in the interior of $P$ are said to cross, whenever their vertices appear alternatingly along the boundary of the polygon. A triangulation $T$ is a maximal plane collection of straight-line edges that connects vertices of $P$. All faces of $T$ except the outer face are triangles. $T$ has $n$ edges between consecutive vertices on the boundary. We call these edges \emph{boundary edges}. Further, there are $n-3$ pair-wise non-crossing \emph{interior edges} which connect non-consecutive vertices of the polygon and split the polygon into $n-2$ triangles. Triangulations of convex polygons admit two natural generalizations, triangulations of point sets in general position and triangulations of polygons (with and without holes). In both settings, triangulations are again maximal plane collections of straight-line edges. For point sets, an edge connects two points in the plane. For polygons an edge connects two vertices of the polygon and is entirely contained in the interior of the polygon. If we require the point set to be in convex position or the polygon to be convex (and hole free), then we recover the original setting of triangulations of convex~polygons.
	
	A \emph{flip} in a triangulation is an operation that takes a convex quadrilateral formed by two adjacent triangles, removes the diagonal that goes through the quadrilateral (and separates the two triangles) and adds the other diagonal of the quadrilateral. For an illustration, see \cref{fig:flip}. The resulting structure is again a triangulation. A \emph{flip sequence} between an \emph{initial} triangulation $T$ and a \emph{target} triangulation~$T'$ is a sequence $T = T_0, T_1, \ldots, T_{\ell-1} , T_\ell = T'$ such that each $T_i$ is a triangulation and two consecutive triangulations only differ by a single flip. The parameter $\ell$ describes the \emph{length} of a flip sequence. The \emph{flip distance} between $T$ and~$T'$ is the minimum $\ell$ for which a flip sequence of length $\ell$ from $T$ to~$T'$ exists. A flip sequence of minimum length~$\ell$ is called a \emph{shortest flip sequence}.
	
	The \emph{flip graph} of triangulations of $P$ is the graph that has all triangulations of $P$ as vertices and an edge between two triangulations if they can be transformed into one another via a single flip. Our definitions can be restated in terms of the flip graph: A flip sequence is a path in the flip graph. The flip distance corresponds to the length of a shortest path in the flip graph. The \emph{diameter} of the flip graph is the maximal flip distance between any pair of triangulations. Since triangulations of a convex polygon are counted by the Catalan numbers, there are $\Theta(\frac{4^n}{n^{3/2}})$ triangulations of a convex $n$-gon. Due to the high number of vertices in the flip graph, it is not feasible to run standard graph algorithms on the flip graph to compute flip distances.
	
	Sleator, Tarjan, and Thurston \cite{sleator1986rotation} showed that any triangulation of a convex $n$-gon can be transformed into any other using at most $2n-10$ flips for $n>12$. This provides an upper bound on the diameter of the flip graph. Further, they showed that the bound on the diameter is tight for infinitely many large enough $n$. For their proof of the lower bound they glue together the starting and target triangulation to form the top and bottom surface of a polyhedron. Then the authors relate the flip distance of the two triangulations to the minimum number of tetrahedra in a $3$-dimensional triangulation of said polyhedron. Afterwards, a lower bound on the number of tetrahedra of such a triangulation is given using a volume argument via hyperbolic geometry.
	Minimizing the number of tetrahedra in a tetrahedralization of a 3-dimensional polytope has been shown to be \NP-hard by Below, de Loera, and Richter-Gebert~\cite{below2000finding}. However, since not every 3-dimensional polytope is obtained by glueing two triangulations along their boundary, this hardness result does not translate back to the flip distance of triangulations.
	Pournin~\cite{pournin2014diameter} later gave a purely combinatorial proof of the diameter of the flip graph by providing families of pairs of triangulations that have distance $2n-10$ from one another. From this proof it follows that the $2n-10$ bound is tight for all $n$ starting from $n=13$.
	
	A question that remained open was the complexity of finding shortest flip sequences between triangulations of convex polygons. Lubiw and Pathak~\cite{lubiw2015flip} and Pilz~\cite{pilz2014flip} independently proved that computing shortest flip sequences between triangulations of point sets is \NP-hard. 
	Both papers also proved NP-hardness of finding shortest flip sequences between triangulations of polygons with holes. 
	Aichholzer, Mulzer, and Pilz~\cite{aichholzer2015flip} proved \NP-hardness for the flip distance problem in simple polygons, getting rid of the necessity of holes. All these hardness proofs are built around counterexamples to the happy edge property. A \emph{happy edge} is an edge that lies in both the initial and target triangulation. The \emph{happy edge property} holds for a reconfiguration problem if there always exists a shortest flip sequence between any pair of configurations that keeps all happy edges in all intermediate configurations. In our case the configurations are triangulations. There exist examples of triangulations of point sets and (simple) polygons for which the happy edge property does not hold~\cite{bose2009flips,hurtado1996flipping}. As mentioned, these were key ingredients in proving \NP-hardness of the computation of shortest flip sequences between triangulations in the mentioned settings.
	 
	In contrast, Sleator, Tarjan, and Thurston~\cite[Lemma~3]{sleator1986rotation} proved that the happy edge property holds for triangulations of convex polygons. Concretely, the authors proved that if two triangulations of a convex polygon have an interior edge in common, then any shortest flip sequence never flips this interior edge and any flip sequence that would flip this interior edge is at least two flips longer than the optimum. Further, they showed that if it is possible to flip one interior edge of $T$ creating a triangulation that has one interior edge more in common with $T'$ than $T$ has, then there exists a shortest flip sequence that starts with this flip. In other words, whenever one can greedily add an edge from the target triangulation, it is optimal to do so. We point out that this does not automatically yield a greedy algorithm, since there are situations where it is not possible to add an edge from the target triangulation directly. Still, these nice properties have not been sufficient to find a polynomial time algorithm that computes the flip distance between triangulations of convex polygons.
	
	We settle the question about the complexity of computing shortest flip sequences between triangulations of convex polygons by proving that it is in fact \NP-hard to compute the flip distance.
	
	\begin{restatable}{theorem}{Main}\label{thm:main}
		Given two triangulations $T$ and $T'$ of the same convex polygon on $n$ vertices and an integer~$k$, deciding whether the flip distance between $T$ and $T'$ is at most $k$ is \NP-complete.
	\end{restatable}
	
	Our proof uses a notion of conflict graphs for flip sequences between triangulations that allows us to encode instances of a planar Max-2SAT variant as pairs of triangulations.
	
	\begin{figure}[ht]
		\centering
		\begin{subfigure}{0.3\textwidth}
			\includegraphics[scale=0.4,page=46]{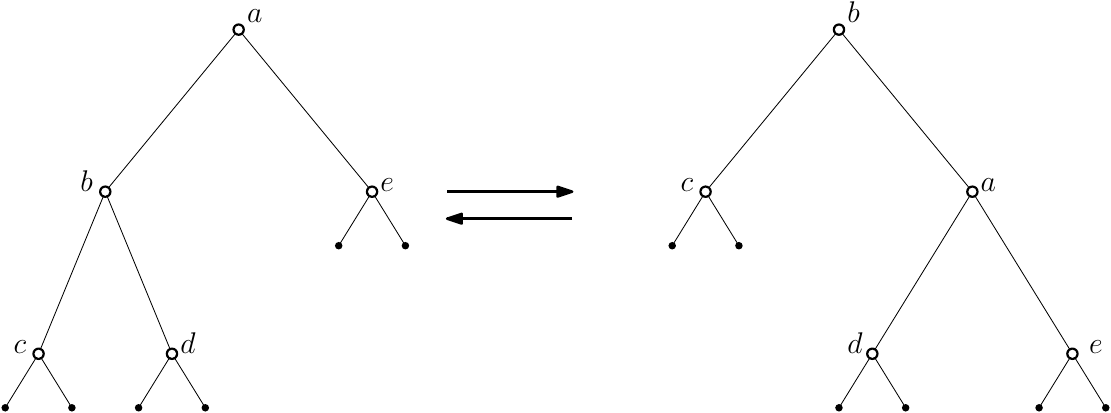}
			\caption{A rotation in a binary tree}
			\label{fig:rotation}
		\end{subfigure}
		\hfill
		\begin{subfigure}{0.4\textwidth}
			\includegraphics[scale=0.5,page=54]{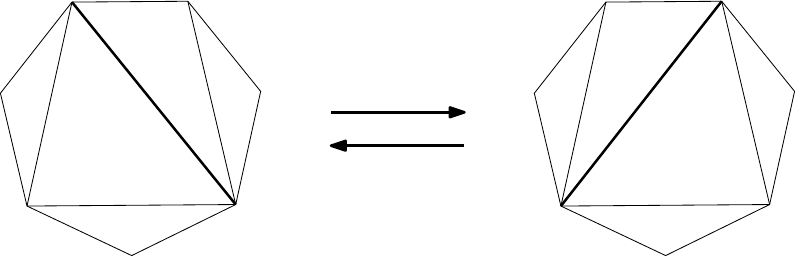}
			\caption{A flip in a triangulation}
			\label{fig:flip}
		\end{subfigure}
		
		\vspace{0.5cm}
		
		\begin{subfigure}{0.8\textwidth}
			\centering
			\includegraphics[scale=0.5,page=47]{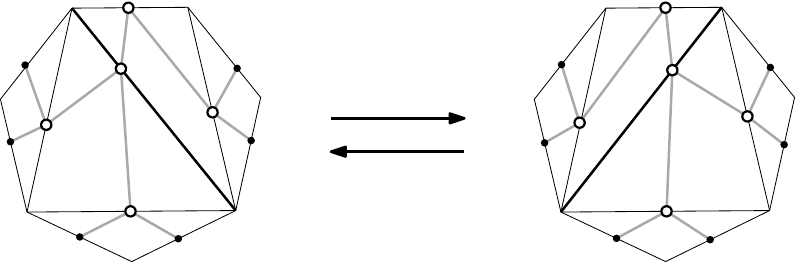}
			\caption{The bijection between triangulations and binary trees preserves flips. The top boundary edge of the triangulation corresponds to the root of the binary tree.}
			\label{fig:fliprot}
		\end{subfigure}
		
		\caption{A rotation in a binary tree and its corresponding flip in a triangulation. First depicted individually and then on top of one another to illustrate the reconfiguration-preserving~bijection.}
		\label{fig:intro}
	\end{figure}
	
	\newpage
	
	\subparagraph{Rotations in Binary Trees.} For a further description of the connection between flips in triangulations of convex polygons and rotations in binary trees we refer to~\cite{sleator1986rotation}. For more connections between Catalan structures, we refer to~\cite{gardner2020time}. A binary tree has two types of nodes, \emph{external} and \emph{interal}. Further, there are three types of relations among nodes, \emph{parent}, \emph{left child}, and \emph{right child}. Every node except for the \emph{root} has a parent. Internal nodes have a left child and a right child, external nodes have no children. A \emph{rotation} in a binary tree is a node change that transforms one binary tree into another: For every internal non-root vertex $b$ and its parent $a$, a rotation swaps $a$ and $b$ in the binary tree. Node $b$ becomes the parent of $a$. If $b$ was the left child of $a$, then $a$ becomes the right child of $b$. Node $b$ keeps its left child, $a$ keeps its right child and the former right child of $b$ becomes the left child of $a$. A symmetric swap can be described if $b$ was the right child. For an illustration of the concepts and the relation to flips in triangulations, see \cref{fig:intro}. The \emph{rotation} graph of binary trees is the graph that has as vertex set all binary trees on $n$ labeled nodes. Two binary trees share an edge if they can be transformed into one another via a single rotation. Rotations are the most commonly used operation when balancing binary trees, that is, creating binary trees where all leaves have roughly the same distance to the root~\cite{knuth1998art,tarjan1983data}. Further, rotations in binary trees provide a framework for studying the space of phylogenetic trees in computational biology. The related local reconfiguration operation is known as nearest-neighbor interchange~(NNI)~\cite{erdHos2021rooted,gambette2017rearrangement}.
	
	The following lemma captures the connection between flips in triangulations of convex polygons and rotations in binary trees.
	
	\begin{lemma}[Lemma 1 in \cite{sleator1986rotation}]\label{lem:binary}
		The rotation graph of binary trees on $n$ nodes is isomorphic to the flip graph of triangulations of convex ($n+2$)-gons.
	\end{lemma}
	
	The bijection that transforms a binary tree into a triangulation and back is depicted in \cref{fig:fliprot}. It is obtained by taking the binary tree as a tree in the dual graph of the triangulation. The nodes of the binary tree are the edges of the triangulation. One boundary edge of the triangulation is a designated root node of the dual graph. Two nodes in the binary tree share an edge if their corresponding edges share a triangle in the triangulation and one edge (representing the parent) is closer to the designated root edge than the other (the child).
	
	From \cref{lem:binary} we immediately conclude an equivalent version of \cref{thm:main} that is of independent~interest.
	
	\begin{theorem}
		Given two binary trees $B$ and $B'$ with $n$ internal nodes and an integer $k$. Deciding whether the rotation distance of $B$ and $B'$ is at most $k$ is \NP-complete.
	\end{theorem}
	
	\section{Further Related Work and Implications}
	
	\subparagraph{Further Algorithmic Results for Flips in Triangulations.} There have been many attempts to find efficient algorithms for the flip distance of triangulations of convex polygons. The search of such algorithms has yielded many sub-results. Among those results is a polynomial time algorithm for finding the flip distance of triangulations of point sets, not necessarily in general position, with no empty five-holes~\cite{eppstein2010happy}. There also exists an upper bound on the flip distance in the number of crossings of the edges of the initial and target triangulation~\cite{hanke1996edge}. While this bound may work well for point sets in general position, it is not applicable for triangulations of convex polygons. It is easy to construct triangulations of linear flip distance, but with quadratically many crossings. For the construction of simple polynomial time computable upper and lower bounds on the flip distance, we refer to~\cite{baril2006efficient} and the related work of this paper. To the best of our knowledge it remains an open problem, whether there exists a polynomial time computable approximation that guarantees an approximation ratio better than $2$. For point sets in general position, however, it is known that such an approximation cannot exist unless the Unique Games Conjecture fails~\cite{pilz2014flip}. Another line of research dealt with obtaining FPT-algorithms: First, there were several simple FPT-algorithms that are a mix of an application of the happy edge property and brute-force~\cite{bosch2023improved,cleary2002restricted,lucas2010improved}. These results were followed up by algorithms that make use of a topological ordering of flips in triangulations. Flips that appear earlier in the topological ordering have to appear earlier in any flip sequence and any flip sequence that respects the topological ordering is a valid flip sequence. It is significantly easier to compute the topological ordering than an actual flip sequence~\cite{feng2021improved,kanj2017computing}. The aforementioned techniques were combined to yield the to date most efficient FPT-algorithm for the flip distance between triangulations of convex point sets~\cite{li20253}.
		
	\subparagraph{Flip Graphs for Plane Graph Reconfiguration Problems.} The notion of a flip graph can also be defined when replacing triangulations with other plane structures.
	
	We could replace triangulations with pseudo-triangulations. A pseudo-triangulation is a collection of pair-wise crossing-free edges such that every inner face has exactly three convex angles. A pseudo-triangulation is called pointed, if every vertex has one angle of at least $180$ degrees. We refer to~\cite{rote2006pseudo} for a survey on pseudo-triangulations. A flip in a (pointed) pseudo-triangulation removes one edge and adds another edge such that the result is again a (pointed) pseudo-triangulation. The flip graph of pseudo-triangulations is connected for any point set~\cite{bronnimann2006counting} and has diameter $O(n\log(n))$~\cite{BEREG2004141}. On convex point sets, the notions of triangulations and pseudo-triangulation coincides. We conclude:
		
	\begin{cor}
		Given two (pointed) pseudo-triangulations $PT$ and $PT'$ and an integer $k$, deciding whether the flip distance from $PT$ to $PT'$ is at most $k$ is \NP-complete.
	\end{cor}
	
	When considering crossing-free perfect matchings, replacing one edge with another will not yield another perfect matching. Instead, we remove two edges and add two edges such that we obtain another crossing-free perfect matching. In convex point sets, the flip distance between any two perfect matchings can be computed in polynomial time~\cite{matchings}. This is interesting, since perfect matchings of convex point sets are also Catalan structures. The number of matchings on a convex $2n$ point set is counted by the $n$-th Catalan number. Even though both triangulations and matchings are counted by Catalan numbers and there exist nice bijections between them, the flip distance problem is \NP-hard for one and polynomial time solvable for the other. Clearly, no bijection between the two structures  commutes with flips and the flip graphs are not isomorphic. When considering point sets in general position, the flip distance problem is \NP-hard for perfect matchings~\cite{MatchingHardness}. The same holds when considering almost perfect matchings, where a flip again removes one edge and adds another edge such that we again get an almost perfect matching. This can be interpreted as the one isolated vertex ``stealing" the edge from another vertex. In point sets in general position, the flip distance problem is \NP-complete~\cite{aichholzer2025flippingoddmatchingsgeometric}. On the contrary, for point sets in convex position the complexity of the flip distance problem remains open.
	
	When considering paths as our underlying structure, the flip distance problem is polynomial time solvable for point sets in convex position~\cite{aichholzer2025linear}. For point sets in general position, the complexity of the flip distance computation is still open. We would, however, be very surprised if the problem is not \NP-hard.
	
	All results on flip distance problems for the reconfiguration of plane graphs that appeared before 2026 followed the pattern that they were easy to solve for convex point sets and \NP-hard for point sets in general position. This pattern was broken when \NP-completeness for the flip distance problem of non-crossing spanning trees was proven, already in the case of convex point sets~\cite{Trees2026hards}.
	
	\subparagraph{The Associahedron and its Generalizations.} We refer to~\cite{ceballos2015many} for a survey on the associahedron. The associahedron $K_n$ is an $(n-2)$-dimensional convex polytope discovered by Stasheff~\cite{stasheff1963homotopy}. The $1$-skeleton of the associahedron is exactly the flip graph of triangulations of a convex polygon with $n+1$ vertices.
	
	The associahedron admits many generalization, for some of which, separate \NP-hardness proofs are known for the respective distance problems. Those include graph associahedra and polymatroids~\cite{ito2023hardness}. These hardness results now also follow as a corollary of \cref{thm:main}.
	Further, the combinatorial distance problem is \NP-hard, and even \APX-hard, for polymatroids that are described by linarly many inequalities~\cite{cardinal2023shortest}.
	
	Conversely, there exists an FPT-algorithm for the rotation distance in graph associahedra~\cite{cunha2025computing} and a polynomial time algorithm for the special case of stellohedra~\cite{CARDINAL2024103877}.
	
	For some other generalizations of the associahedron, NP-hardness of computing the distance of two elements when walking along the 1-skeleton follows as a corollary of our result.
	
	\begin{cor}
		Given two elements that each describe a vertex of the following polytope, computing the distance of the two elements in the $1$-skeleton of the polytope is \NP-hard:
		\begin{itemize}
			\setlength\itemsep{0em}
			\item Generalized associahedra and cluster algebras~\cite{MR1941227,MR1887642}.
			\item Brick polytopes~\cite{MR2820759,MR3327085}.
			\item Generalized permutahedra~\cite{MR2487491}.
			\item Quotientopes ~\cite{MR3964495}.
			\item The Mutliassociahedron~\cite{MR3221537,MR2167478}.
			\item Operahedra~\cite{LAPLANTEANFOSSI2022108494}
		\end{itemize}
	\end{cor}
	
	Further, there are dual graphs of spherical subword complexes~\cite{MR3144391,MR2047852} which are conjectured to be realizable as polytopes, but no such realization is known. Computing the distance between two of their vertices is also \NP-hard, as they generalize the associahedron.
	
	Moreover, there are polyhedra whose $1$-skeleton contains the flip graph of triangulations of convex polygons as a subgraph such as cosmohedra~\cite{Arkani_Hamed_2025} or permuto-associahedra~\cite{reiner1994coxeter}. However, it is not trivially clear that any shortest path between two vertices in this subgraph will only walk along vertices within the subgraph. There might be shorter paths that leave and re-enter the subgraph. Therefore, the complexity of computation of shortest combinatorial paths remains an open question for these polytopes. To prove \NP-hardness it would be sufficient to prove that shortest paths between vertices in the subgraph that is isomorphic to the flip graph of triangulations of convex polygons never leave this subgraph. Analogs to this property are known as strong convexity or non-leaving face property and have been proven for similar structures~\cite{ceballos2016diameter,williams2017w}.
	
	\subparagraph{Discrete Optimization.} In 2022, De Loera, Kafer, and Sanità~\cite{de2022pivot} proved that computing shortest monotone paths in polytopes is \NP-hard. As a consequence, any efficiently computable pivot rule for the simplex method cannot be guaranteed to reach an optimal solution of a linear program in the minimum number of non-degenerate pivots, unless \P=\NP. The polytopes that appear in their proof are degenerate, that is, their vertex-edge graph and their feasible basis exchange graph do not coincide. The authors of~\cite{de2022pivot} therefore motivate the question about the computation of shortest paths and shortest monotone paths in non-degenerate, or \emph{simple}, polytopes. The associahedron is such a simple polytope (for example see~\cite{devadoss2009realization}) and it can be described using $\Theta(n^2)$ linear equations. As a corollary of~\cref{thm:main} we obtain.
	
	\begin{cor}
		Given a simple polytope $\mathcal{P}$, described by linear inequalities, two vertices $x,y$ of~$\mathcal{P}$, and an integer $k$. Deciding whether $x$ and $y$ have a distance of at most $k$ in the graph of $\mathcal{P}$ is \NP-hard.
	\end{cor}
	
	In parallel to our work, Black and Steiner~\cite{black2026finding} also proved that finding shortest paths in simple polytopes is \NP-hard by a clever construction of a polytope that is associated to the \textsc{Partition with even sum} problem. Moreover, from their proof it also follows that computing shortest monotone paths in simple polytopes is \NP-hard. Consequently, for LP's for simple polytopes there is no efficiently computable pivot rule for the simplex method that can always reach an optimal solution in the minimum number of total pivots unless \P=\NP.
	
	\subparagraph{The Tamari Lattice and its generalizations.} The Tamari Lattice has been introduced by Tamari~\cite{tamari1962algebra}. In general, a lattice is a partially ordered set with a \emph{join} relation and a \emph{meet} relation. For every pair of elements $x$ and $y$ in the lattice there exists a join $x\lor y$ that is the unique least upper bound of $x$ and $y$. There further exists a meet $x\land y$ that is the unique greatest common lower bound of $x$ and $y$. For the Tamari Lattice, this cover relation has first been described in terms of bracketings, defined by $(((ab)c)d) < ((ab)(cd)) < (a(b(cd)))$ and $ (((ab)c)d) < ((a(bc))d) < (a((bc)d)) < (a(b(cd)))$. Bracketings are another well known strucure that is counted by the Catalan numbers. There exists a bijection between bracketings and triangulations such that the cover relation translates to triangulations. One triangulation covers another triangulation if they differ by a single flip and the edge that is different is ''steeper" with respect to some reference line in the covering triangulation compared to the covered triangulation.
	
	Many generalizations of the Tamari Lattice have been introduced in the past. Some of them are nicely surveyed in~\cite{vonbell2025framinglatticesflowpolytopes}. We conclude that finding shortest paths that walk along cover relations in any such generalization is \NP-hard.
	
	\begin{cor}
		Given two elements of a lattice, computing the shortest path along cover relations is \NP-hard in the following lattices:
		\begin{itemize}
			\setlength\itemsep{0em}
			\item The Tamari lattice~\cite{tamari1962algebra}.
			\item Framing lattices by von Bell and Ceballos~\cite{vonbell2025framinglatticesflowpolytopes}.
			\item The $\nu$-Tamari lattices of Préville-Ratelle and Viennot~\cite{preville2017enumeration}.
			\item The alt $\nu$-Tamari lattices of Ceballos and Chenevière~\cite{ceballos2023linear}.
			\item The Type A Cambrian lattice of Reading~\cite{reading2006cambrian}.
			\item The $(\epsilon,I,\overline{J})$-cambrian lattices by Pilaud~\cite{pilaud2020cambrian}.
			\item The Permutree lattice by Pilaud and Pons~\cite{pilaud2018permutrees}.
			\item The Grassmann-Tamari order of Santos, Stump, and Welker~\cite{santos2017noncrossing}.
			\item The Grid-Tamari lattices of McConville~\cite{mcconville2017lattice}.
		\end{itemize}
	\end{cor}
	
	\section{Preliminaries: Linear Representation, Blow-ups, and Conflict Graphs}\label{sec:prelim}
	
	In this section we introduce some important tools for our proof. We point out that counterparts of these tools have been developed for the study of the diameter of the flip graph of non-crossing spanning trees~\cite{bjerkevik2024flippingnoncrossingspanningtrees}. We adapt them for the study of flip sequences between triangulations. To maintain consistency with previous work~\cite{Trees2026hards,bjerkevik2024flippingnoncrossingspanningtrees}, we adopt similar notation wherever possible.
	
	\subparagraph{Linear Representation.} A usual way to represent a convex $n$-point polygon and a triangulation of it is by placing $n$ points equidistantly on a circle, with vertices numbered in counterclockwise order. Edges are then straight lines between points. 
	 
	For our approach we use a different way to depict the triangulation, namely a \emph{linear representation} as depicted in \Cref{fig:linear}. In a linear representation, the vertices are placed on a horizontal line called the \emph{spine}. Edges of the triangulation are represented by semicircles above the spine. An exception is made for the $n-1$ edges between consecutive vertices on the spine which are represented by straight lines. This can be interpreted as cutting and unfolding the circle on which the vertices were previously placed.
	
	\begin{figure}[t]
		\centering
		\includegraphics[scale=0.65,page=1]{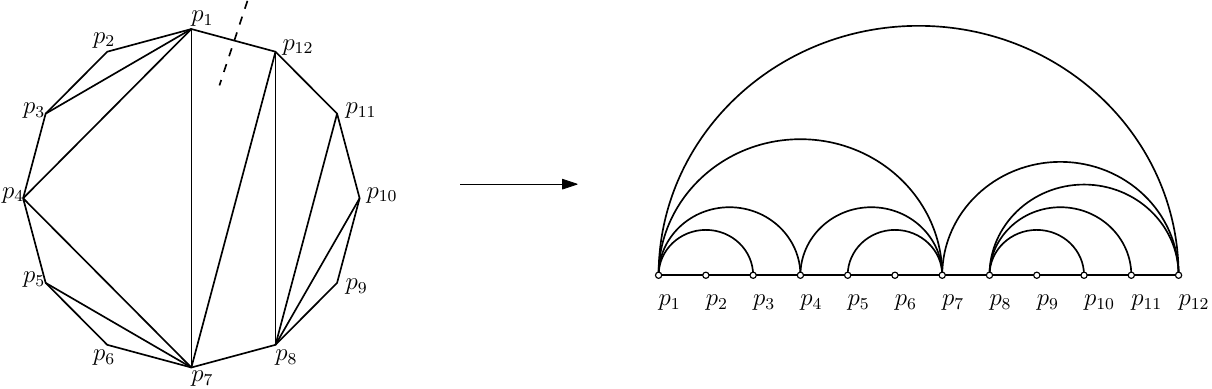}
		\caption{Left: A triangulation of a regular 12-gon. Right: The same triangulation in a linear representation.}
		\label{fig:linear}
	\end{figure}
	
	When working with flip sequences between two triangulations $T$ and $T'$, we draw $T$ on top of the spine and $T'$ below the spine, such that the two triangulations share the same points on the spine that represent the vertices of the polygon. For an illustration, see \Cref{fig:twotriangles}.
	
	The spine induces a natural ordering of its vertices by saying $p_i<p_j$ if the vertex $p_i$ lies further to the left on the spine than vertex $p_j$.
	
	\subparagraph{Blow-ups.}
	
	Let $\Gamma$ be the set of all pairs of triangles $(\Delta,\Delta')$ where $\Delta$ is a triangle in $T$ and $\Delta'$ a triangle in $T'$ such that $\Delta$ and $\Delta'$ have exactly one edge $e$ in common and this common edge is an edge that joins two consecutive spine vertices.
	
	To construct the $\beta$-\emph{blow-up} of the triangulations $T$ and $T'$, we apply the following procedure to every pair $(\Delta,\Delta') \in \Gamma$. We subdivide $e$ into $\beta+1$ edges placing $\beta$ additional vertices, one at every division point. We write $V(\Delta)$ to denote the set of newly introduced vertices at the division points of $e$ together with the original vertices of $e$. Further, we write $V^O(\Delta)$, if we refer to $V(\Delta)$, but without the vertices of $e$.
	
	Let $v$ be the vertex of $\Delta$ that is not incident to $e$. To again obtain a triangulation of the new polygon, we add $\beta$ new edges. One from every vertex in $V^O(\Delta)$ to $v$. We write $\Lambda(\Delta)$ to denote the set of all edges from $v$ to $V^O(\Delta)$. We call $\Lambda(\Delta)$ the \emph{fan} of $\Delta$.
	
	We apply the same procedure to $\Delta'$ in $T'$ and define the fan $\Lambda(\Delta')$ of $\Delta'$ analogously. For an illustration of a $2$-blow-up, see \Cref{fig:blowup}.
	
	The triangulations that result from applying a $\beta$-blow up are denoted by $\beta T$ and $\beta T'$. These triangulations then both have $n + \beta \lvert \Gamma \rvert$ vertices, where $n$ denotes the number of vertices of the triangulations $T$ and $T'$. This may, for large values of $\beta$, be significantly larger than just $n$. Further, the triangulations $\beta T$ and $\beta T'$ each have $(n-3) + \beta \lvert \Gamma \rvert$ interior edges.
	
	\begin{figure}[htb]
		\centering
    	\begin{subfigure}{0.48\textwidth}
    		\centering
			\includegraphics[scale=0.6,page=2]{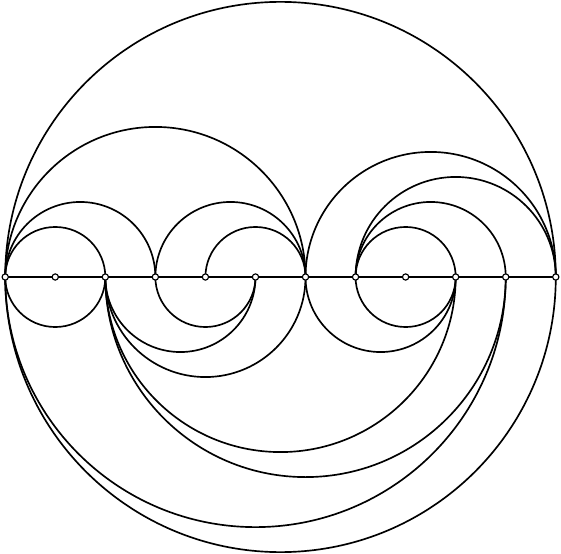}
			\caption{Two triangulations in the same linear representation}
			\label{fig:twotriangles}
		\end{subfigure}
		\hfill
		\begin{subfigure}{0.48\textwidth}
			\centering
			\includegraphics[scale=0.6,page=3]{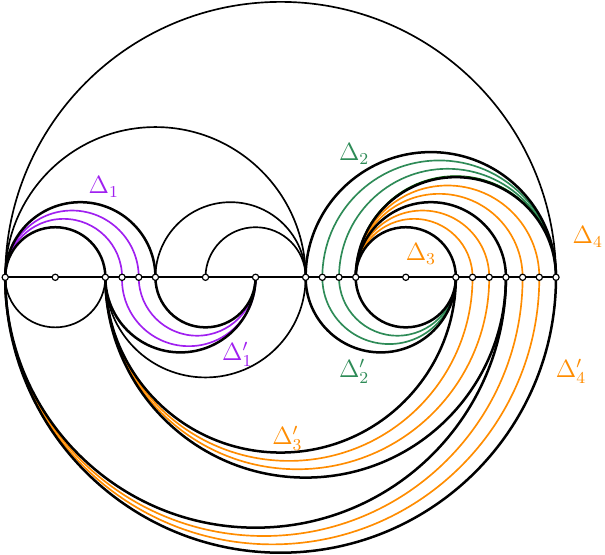}
			\caption{Two triangulations in a linear representation with a 2-blow-up}
			\label{fig:blowup}
		\end{subfigure}
		\caption{Further concepts involving the linear representation}
	\end{figure}
	
	\subparagraph{Conflict graphs.} We introduce notation to capture when edges of a blow-up triangulation~$\beta T$ need to be removed in a flip sequence before new edges of a blow-up triangulation $\beta T'$ can be added via flips.
	
	The \emph{conflict graph} $H=H(T,T')$ is a directed graph that has vertex set $\Gamma$ and a directed edge from a pair $(\Delta_i,\Delta'_i)$ to a pair $(\Delta_j,\Delta'_j)$ if, after applying a blow up to the triangulations, one edge in $\Lambda(\Delta_i)$ (and hence all edges in $\Lambda(\Delta_i)$) crosses an edge (and hence all edges) in $\Lambda(\Delta'_j)$. A directed edge in the conflict graph represents the fact that we cannot have edges of $\Lambda(\Delta'_j)$  in the triangulation as long as we still have edges from $\Lambda(\Delta_i)$ as part of our triangulation.
	
	If we consider the illustration from \Cref{fig:blowup}, we see that a conflict graph has a directed edge from $(\Delta_1,\Delta'_1)$ to $(\Delta_3,\Delta'_3)$. Further, there is a bi-directed edge between $(\Delta_2,\Delta'_2)$ and $(\Delta_3,\Delta'_3)$, as well as  between $(\Delta_2,\Delta'_2)$ and  $(\Delta_4,\Delta'_4)$. Finally, there is also a directed edge from  $(\Delta_4,\Delta'_4)$ to $(\Delta_3,\Delta'_3)$.
			
	\subparagraph{Above, below and crossing pairs.} For our upcoming \NP-hardness reduction, we will only use three \emph{types of triangle pairs} which are depicted in \cref{fig:pairs}. For a pair of triangles $(\Delta_i,\Delta'_i)$ We denote the triangle of~$T$ by $v_{i,1}v_{i,2}v_{i,3}$ and the triangle of~$T'$ by $v_{i,1}v'_{i,2}v_{i,3}$. Note that we use the letter $v_i$ to denote points now instead of $p_i$ in \cref{fig:linear}. We used the indices for the $p$'s to count the vertices in the order along the spine, whereas now the index for $v$'s denote which pair of triangles they belong to which may not be tied to an order along the spine. We call a triangle pair \emph{above} and color it in green if $v_{i,2}<v'_{i,2}<v_{i,1}<v_{i,3}$. A triangle pair will be colored in orange and called \emph{below} if $v_{i,1}<v_{i,3}<v_{i,2}<v'_{i,2}$. The last type of triangle pairs will be called crossing and colored in purple if $v'_{i,2}<v_{i,1}<v_{i,3}<v_{i,2}$. We remark that for all three types of triangle pairs there also exists a mirrored counterpart which can be observed in \cref{fig:blowup}. These will, however, not show up in our construction and do not get a name. Further, above, below, and crossing pairs together with their mirrored counterparts make up all possible triangle pairs.
	
	\begin{figure}[htb]
		\centering
		\begin{subfigure}{0.3\textwidth}
			\centering
			\includegraphics[scale=0.7,page=43]{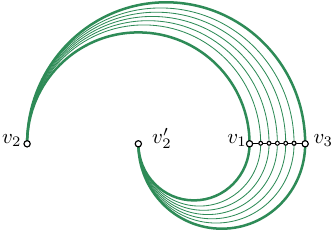}
			\caption{An above pair.}
		\end{subfigure}
		\hfill
		\begin{subfigure}{0.3\textwidth}
			\centering
			\includegraphics[scale=0.7,page=44]{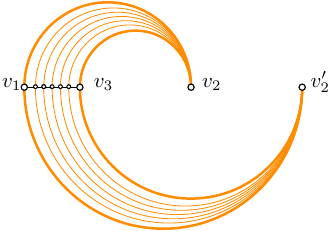}
			\caption{A below pair.}
		\end{subfigure}
		\hfill
		\begin{subfigure}{0.3\textwidth}
			\centering
			\includegraphics[scale=0.7,page=45]{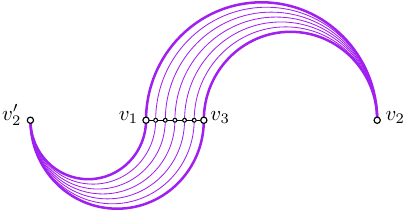}
			\caption{A crossing pair.}
		\end{subfigure}
		\caption{The three types of triangle pairs.}
		\label{fig:pairs}
	\end{figure}
	
	We will now investigate what conflicts between triangle pairs of the same type look like.
	
	\begin{lemma}\label{lem:abc}
		For two triangle pairs $(\Delta_i,\Delta'_i)$ and $(\Delta_j,\Delta'_j)$ with edges $e_i$ and $e_j$ on the boundary, we observe the following conflicts: If $(\Delta_i,\Delta'_i)$ and $(\Delta_j,\Delta'_j)$ are both pairs of the same type (above, below, or crossing) and we have an edge $(\Delta_i,\Delta'_i)\to(\Delta_j,\Delta'_j)$ in the conflict graph, then the spine edge of $(\Delta_i,\Delta'_i)$, denoted by $e_i$, is further to the left than the spin edge of $(\Delta_j,\Delta'_j)$, $e_j$, in the linear representation.
	\end{lemma}
	
	\begin{proof}
		We denote $\Delta_i = v_{i,1}v_{i,2}v_{i,3}$, $\Delta'_i = v_{i,1}v'_{i,2}v_{i,3}$, $\Delta_j = v_{j,1}v_{j,2}v_{j,3}$, and $\Delta'_j = v_{j,1}v'_{j,2}v_{j,3}$.
		
		First assume, both pairs are above pairs. Then we study edges of the form $v_{i,2}v_i$ with $v_{i,1}<v_i<v_{i,3}$ that cross edges of the form $v'_{j,2}v_j$ with $v_{j,1}<v_j<v_{j,3}$. For this we need $v_{i,2} < v'_{j,2} < v_i < v_j$ and consequently $v_{i,3}<v_{j,1}$ hence $e_i$ is further to the left than $e_j$.
		
		Second, let both pairs be below pairs. The statement follows analogously as for above pairs if we rotate the linear representation by 180 degrees.
		
		Finally, consider two crossing pairs. We require that the edges $v'_{j,2}v_j$ with $v_{j,1}<v_j<v_{j,3}$ are crossed by edges $v_iv_{i,2}$ with $v_{i_1} < v_i<v_{i,3}$. This can happen in two forms. Either $v'_{j,2} < v_i < v_j < v_{i,2}$ or $v_i < v'_{j,2} < v_{i,2} < v_j$. In both cases we conclude that $v_i<v_j$ and therefore $e_i$ is further to the left than $e_j$.
	\end{proof}
	
	\cref{lem:abc} will be helpful when we try to prove acyclicity of subsets of the conflict graph. We will use it in the following form.
	
	\begin{lemma}\label{lem:acyclic}
		Given a subset $S$ of the vertices of the conflict graph $H=H(T,T')$ that has the following two properties.
		\begin{itemize}
						\setlength\itemsep{0em}
			\item Every vertex of $S$ represents a pair of triangles of one of the three types (above, below, crossing) that we introduced in this section.
			\item Every edge in $S$ are either between two pairs of the same type, or they are directed from an above pair to a crossing pair, from an above pair to a below pair, or from a crossing pair to a below pair.
		\end{itemize}
		Then $S$ forms an acyclic subset.
	\end{lemma}
	
	\begin{proof}
		By \cref{lem:abc}, $S$ does not induce a directed cycle that consist entirely of pairs of triangles of a single type. By the second condition of the lemma, we have no directed paths from crossing pairs to above pairs or from below pairs to crossing or above pairs. In particular, we have no directed cycles that contain pairs of triangles of different types in $S$. We conclude that $S$ has no cycles at all.
	\end{proof}
	
	\cref{lem:acyclic} is the crucial tool when arguing that a given set $S$ is acyclic. Instead of making big tedious checks to ensure acyclicity, we will only have to check that $S$ fulfills the second condition of the lemma. As mentioned, for our reduction we will only use pairs of triangles that are of the three types, above, below, and crossing. Therefore, we get the first condition for free.
	
	\section{Outline of the Proof}\label{sec:puttogether}
	
	Our proof of \cref{thm:main} goes via an intermediate problem. First, we will proof in \cref{sec:hard} that computing maximum acyclic subsets in the conflict graph of two triangulations with a given linear representation is \NP-hard. The proof reduces from a variant of \textsc{planar Max-2SAT}.
	
	\begin{restatable}{proposition}{achard}\label{prop:achard}
		Given  triangulations $T$ and $T'$ with a linear representation and corresponding conflict graph $H$, as well as an integer $k$, it is \NP-complete to decide whether $H$ has an acyclic subset of size at least $k$.
	\end{restatable}

	Then, to show our main results, we will show that, if we apply $\beta$-blow-ups for large values of $\beta$, then the flip distance between the two triangulations $\beta T$ and $\beta T'$ depends, up to small deviations, only on the size of the largest acyclic subset of the conflict graph $H=H(T,T')$. We will denote the size of the largest acyclic subset of $H$ by $ac(H)$.

	In \cref{sec:flipsequence}, we will derive a conflict-graph-based upper bound on the flip distance of $\beta T$ and~$\beta T'$. Let $S$ be an acyclic subset of size $ac(H)$. Our upper bound is constructive and produces a flip sequence that treats triangle pairs in $S$ differently from triangle pairs that are not in $S$. If $(\Delta,\Delta')$ is not part of $S$, then we will invest $2\beta$ flips, one flip for every edge in $\Lambda(\Delta)$ and one flip for every edge in $\Lambda(\Delta')$.	After applying those flips, we end up with a set of $\beta$ edges that are contained in both triangulations. These flips are meant for getting edges out of the way to create room for the next sequence of flips.
	
	For every pair $(\Delta,\Delta')$ in $S$, we will invest a small number of set-up flips where the number only depends on the size of $T$ and $T'$, but not on $\beta$. These flips allow us to then insert all the edges of $\Lambda(\Delta')$ using one flip per edge. Combining everything gives an upper bound on the flip distance.
	
	\begin{restatable}{proposition}{upper}\label{prop:upper}
		Given two triangulations $T$ and $T'$ with a linear representation. Let $H=H(T,T')$ be the conflict graph of $T$ and $T'$ and $ac(H)$ the size of its largest acyclic subset. The flip distance between $\beta T$ and $\beta T'$ is at most $\beta (2\lvert \Gamma \rvert - ac(H)) + n(2n-5)$.
	\end{restatable}
	
	\cref{sec:lowerbound} is dedicated to finding a matching lower bound. For a flip sequence $F$, we distinguish two sets of pairs of triangles $(\Delta,\Delta')$. A pair will be called direct, if some triangulation in the flip sequence already contains edges of $\Lambda(\Delta')$ while it still contains edges of $\Lambda(\Delta)$ and indirect otherwise. We will then bound the number of direct pairs in a flip sequence by the cardinality of the largest acyclic subset of the conflict graph. For any indirect pair, we will show that the flip sequence needs to introduce many edges that are neither part of $\beta T$ nor of $\beta T'$. These edges introduce additional flips that need to be made in the flip sequence, allowing us to make estimates on the minimum number of flips in a flip sequence. The crucial part of the proof is showing that for distinct indirect pairs the sets of introduced edges are reasonably disjoint which implies that the number of additional flips behaves in an additive way.
	
	\begin{restatable}{proposition}{lower}\label{prop:lower}
		Given two triangulations $T$ and $T'$ with a linear representation. Let $H=H(T,T')$ be the conflict graph of $T$ and $T'$ and $ac(H)$ the size of its largest acyclic subset. The flip distance between $\beta T$ and $\beta T'$ is at least $\beta (2\lvert\Gamma\rvert - ac(H)) - 3n^2$.
	\end{restatable}
	
	We want to point out that the obtained bounds are far from optimal for small values of $\beta$. If $\beta$ is of roughly the same order or smaller than $n$, then the lower bound on the flip distance is negative and the upper bound is quadratic in the size of the triangulations. However, we know that there exist upper bounds on the flip distance that are linear in the size of the triangulations. For large values of $\beta$, at least quadratic in $n$, the terms in $\beta$ dominate the flip distance and the size of the largest acyclic set determines the flip distance between $\beta T$ and $\beta T'$ up to ``rounding errors" in $n$.
	
	We combine the three partial results to obtain a final proof for \cref{thm:main}.
	
	\Main*
	
	\begin{proof}
		Containment in \NP~follows from the fact that a shortest flip sequence between $T$ and~$T'$ has length that is linear in the size of $T$ and $T'$ and each occurring triangulation has linear size. Therefore, a flip sequence of length at most $k$ serves as a certificate that can be checked in time and space that is polynomial in the input size.
		
		By \cref{prop:achard} it is \NP-hard to check whether a conflict graph $H$ with $\Gamma$ as the set of triangle pairs of triangulations $T$ and $T'$ on $n$ vertices with a given linear representation has an acyclic subset of size at least $k$.
		
		We now consider the $\beta = 6(n^2+n)$-blow-up of $T$ and $T'$, denoted by $\beta T$ and $\beta T'$. Since the size of the triangulations increases by a polynomial factor, this reduction step is polynomial.
		
		To obtain the statement of the theorem, we show that the following properties are equivalent. In particular, if one property is \NP-hard to decide, then so is the other.
		\begin{itemize}
			\item $H$ has an acyclic subset of size at least $k$.
			\item The flip distance between $T$ and $T'$ is at most $\beta (2\lvert\Gamma\rvert - k) +n(2n-5)$.
		\end{itemize}
		
		One direction is the statement of \cref{prop:upper}. Assume for the other direction that~$H$ does not contain an acyclic set of size $k$ or larger. Then, by \cref{prop:lower}, we obtain a lower bound on the flip distance of $\beta T$ and $\beta T'$.
		
		\begin{equation*}
			\beta(2\lvert \Gamma \rvert - (k-1)) - 3n^2 = \beta(2\lvert \Gamma \rvert - k) + \beta - 3n^2 > \beta (2\lvert\Gamma\rvert - k) +n(2n-5)
		\end{equation*}
		
		The last inequality follows since $\beta = 6(n^2+n) - 3n^2 > 2n^2 - 5n$ and concludes the proof.
	\end{proof}

	\section{Hardness of Finding Maximum Acyclic Subsets}\label{sec:hard}
	
	Our proof that finding shortest flip sequences between two triangulations of a convex point set is \NP-hard goes via an intermediate problem. First, we show that it is \NP-hard to compute maximum acyclic subsets in a directed graph $G$, even if $G$ is the conflict graph of two triangulations with a given linear representation. A similar trick was used to prove that finding shortest flip sequences between non-crossing spanning trees is \NP-hard~\cite{Trees2026hards}. We, however, use a simplified conflict graph with simpler variable gadgets.
	
	\achard*
	
	The problem is contained in \NP\ and a valid acyclic set serves as a certificate. Any subset of~$\Gamma$ is of polynomial size in the input. Checking that the set is indeed acyclic can be done using only polynomial space and time.
		
	To prove \NP-hardness we reduce from the \NP-hard problem \textsc{Planar Separable Monotone Max-2Sat}~\cite{buchin2019geometric}. For \textsc{Planar Separable Monotone Max-2SAT} we are given a \textsc{2SAT} formula $\phi$ with a set $X$ of variables and a set $C$ of clauses that each contain two literals. The vertex-clause incidence graph $G_\phi$ admits an \emph{aligned drawing} $D(G_\phi)$, that is, a drawing that has all variables in $X$ along a horizontal line $L$. Every clause in $C$ lies entirely below or above $L$ and no edge of $G_\phi$ crosses $L$ or any other edge in $D(G_\phi)$. Any clause that lies in the upper half plane only contains positive literals. Any clause in the lower half plane contains only negative literals. Additionally, we are given an integer $k'$. \textsc{Planar Separable Monotone Max-2SAT} asks, whether we can assign truth values to $X$ such that at least $k'$ clauses in $C$ are satisfied.
	
	Our reduction from \textsc{Planar Separable Monotone Max-2SAT} to finding maximum acyclic subsets in the conflict graph utilizes small gadgets. These gadgets will be combined into two trianglations $T$ and $T'$ in a linear representation such that for the resulting conflict graph $H$ the number of vertices in an acyclic subset of $H$ corresponds to clauses in $\phi$ that can be fulfilled. We remark that we construct fewer triangles than we would need to form a full triangulation. The final triangulations $T$ and $T'$ will be obtained by triangulating bigger faces in the interior arbitrarily. We continue by describing the construction of each gadget in more detail. At the same time, we describe how we construct a set $S$ which we will prove to be acyclic.
	
	\begin{figure}[h]
		\centering
		\begin{subfigure}{0.4\textwidth}
			\centering
			\includegraphics[scale=0.8,page=5]{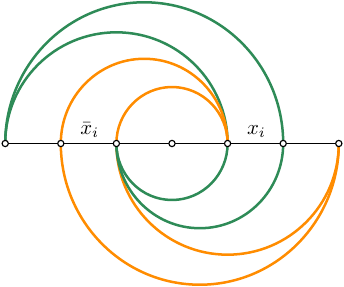}
			\caption{The set up for the variable gadget.}
			\label{fig:variable}
		\end{subfigure}
		\hfill
		\begin{subfigure}{0.4\textwidth}
			\centering
			\includegraphics[scale=0.8,page=6]{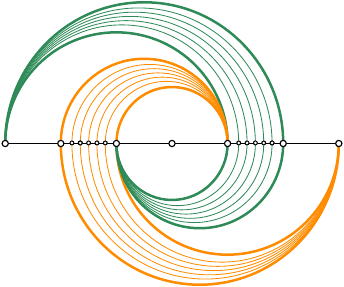}
			\caption{The full variable gadget with multiple triangles.}
			\label{fig:variableup}
		\end{subfigure}
		\begin{subfigure}{0.4\textwidth}
			\centering
			\includegraphics[scale=0.8,page=7]{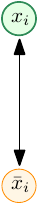}
			\caption{Conflict graph of the set up of the variable gadget.}
			\label{fig:vconflict}
		\end{subfigure}
		\hfill
		\begin{subfigure}{0.4\textwidth}
			\centering
			\includegraphics[scale=0.8,page=8]{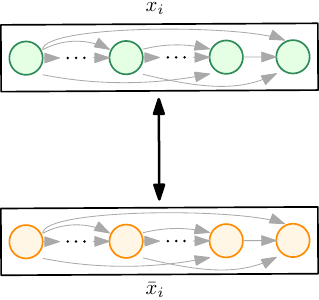}
			\caption{Conflict graph of the full variable gadget with multiple triangles.}
			\label{fig:vconflictup}
		\end{subfigure}
		\caption{The variable gadget and its conflict graph.}
		\label{fig:wholevariable}
	\end{figure}
	
	\subparagraph{Variable gadget.} For every variable $x_i$ we introduce a variable gadget as depicted in \cref{fig:wholevariable}. Its structure is based on two pairs of triangles denoted by $x_i$ and $\bar{x}_i$ and represented by a green above pair, resp. an orange below pair. Let $\Delta_{x_i}$ and $\Delta'_{x_i}$ as well as $\Delta_{\bar{x}_i}$ and $\Delta'_{\bar{x}_i}$ denote the individual triangles of the triangle pairs. Observe that after performing a blow-up, the edges of $\Lambda(\Delta_{x_i})$ cross the edges of $\Lambda(\Delta'_{\bar{x}_i})$ and the edges of $\Lambda(\Delta_{\bar{x}_i})$ cross the edges of $\Lambda(\Delta'_{x_i})$. Therefore, we have conflicts $x_i\to\bar{x}_i$ and $\bar{x}_i\to x_i$ in the conflict graph. See \cref{fig:variable} and \cref{fig:vconflict} for an illustration of the set up.
	
	An acyclic subset of the conflict graph can now contain $x_i$ or $\bar{x}_i$, but not both. The choice, which of the two pairs is chosen, corresponds to the truth assignment of the variable $x_i$.
	
	To obtain the final variable gadget, we need to perform an additional construction step. The reason for this is that we will later connect variable gadgets to the clause gadgets. This will lead to additional conflicts in the conflict graph. We want to make sure that the truth assignment of the variables is well-defined. This well-definedness could be at risk, if an acyclic subset contains neither of $x_i$ or $\bar{x}_i$ and can, in return, contain more triangle pairs from clause gadgets. To counter this, we not only introduce just a single triangle for $x_i$ and $\bar{x}_i$, but $\lvert C\rvert+1$ of them. To achieve this, like in the blow-up operation, we subdivide the edge that the two triangles of a pair share with the boundary by placing $\lvert C \rvert$ vertices on it and adding edges from all the introduced vertices to the vertex that is not incident to the boundary edge, see \cref{fig:variableup}.
	
	The resulting conflict graph is depicted in \cref{fig:vconflictup}. We observe that all pairs of triangles that result from subdividing $x_i$ have a bidirected conflict with all pairs of triangles that result from subdividing $\bar{x}_i$. For better visibility, we collect all these bidirected conflicts by one single vertical edge in \cref{fig:vconflictup}. All other new conflicts that arise are between above pairs or between below pairs and are in accordance with \cref{lem:acyclic}. We collect the main insight about the variable gadget in the following lemma.
	
	\begin{lemma}
		The conflict graph for the variable gadget has exactly two maximum acyclic subsets of cardinality $\lvert C \rvert+1$ given by all pairs of triangles that arise from subdividing $x_i$, respectively~$\bar{x}_i$.
	\end{lemma}
	
	To obtain an acyclic set $S$ from a truth assignment we will add all the pairs labeled $x_i$ to $S$ whenever $x_i$ has value \textsc{True} assigned to it. Otherwise, we add pairs labeled $\bar{x}_i$ to $S$. Conversely, from a large acyclic subset $S$, we reconstruct a truth assignment for the variables by assigning \textsc{True} to $x_i$ if $S$ contains triangle pairs labeled $x_i$. Otherwise, we assign \textsc{False}.
	
	\subparagraph{Positive clause gadgets.} Consider a clause from the upper half plane that contains two positive literals $x_i\lor x_j$ with $i<j$. We add a clause gadget that consists of the two pairs of triangles $C_1$ and $C_2$ as depicted in \cref{fig:clausepos}. For the conflict graph in \cref{fig:cconflictpos} we use a single vertex for triangle pairs $x_i$ and $\bar{x}_i$ (resp. $x_j$ and $\bar{x}_j$). All conflicts that we observe hold for every pair of triangles that resulted from subdividing the original triangles. We observe that we obtain three bidirected conflicts: between $\bar{x}_i$ and $C_1$, $\bar{x}_j$ and $C_2$, and between $C_1$ and $C_2$.
	
	\begin{figure}[htb]
		\centering
		\begin{subfigure}{0.4\textwidth}
			\centering
			\includegraphics[scale=0.4,page=9]{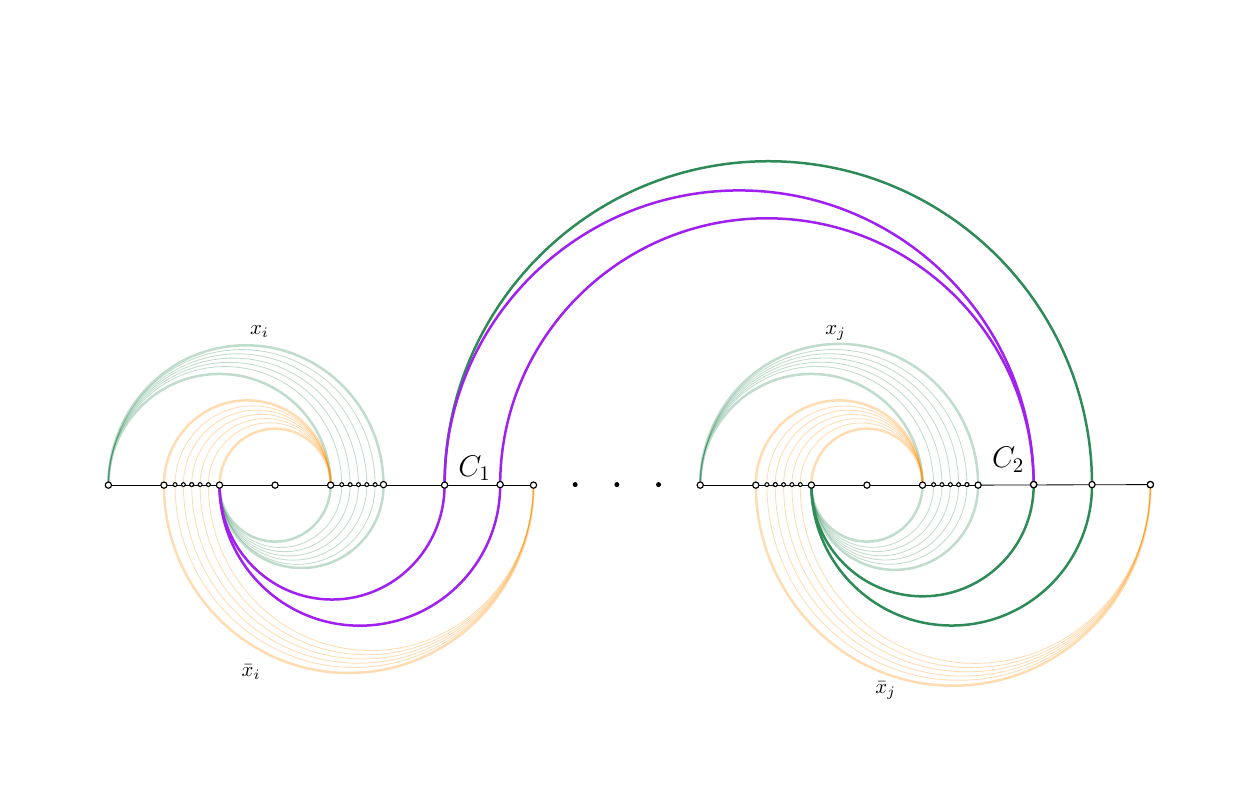}
			\caption{The clause gadget for a positive clause together with the relevant variable gadgets.}
			\label{fig:clausepos}
		\end{subfigure}
		\hfill
		\begin{subfigure}{0.4\textwidth}
			\centering
			\includegraphics[scale=0.7,page=11]{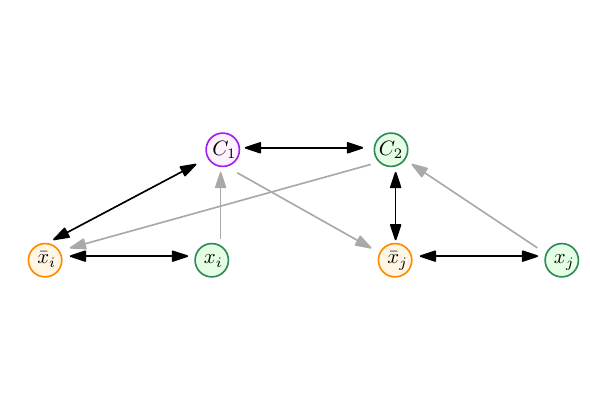}
			\caption{Conflict graph of a clause gadget with positive literals.}
			\label{fig:cconflictpos}
		\end{subfigure}
		\caption{The positive clause gadget and its conflict graph.}
	\end{figure}

	\begin{figure}[htb]
		\centering
		\begin{subfigure}{0.4\textwidth}
			\centering
			\includegraphics[scale=0.4,page=31]{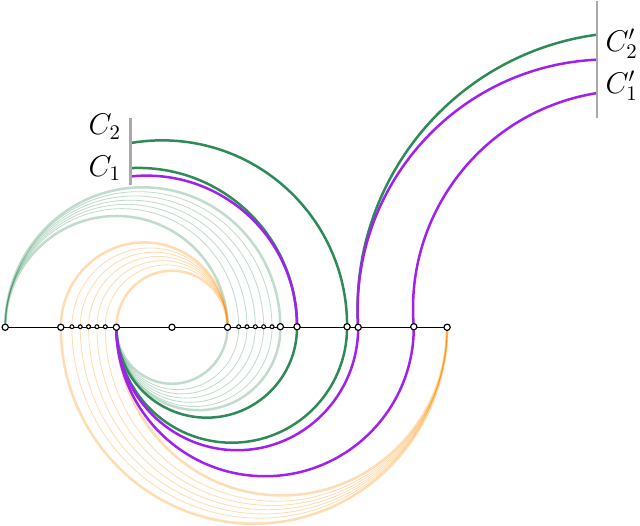}
			\caption{Multiple positive clause gadgets extending to different directions.}
			\label{fig:clausepos1}
		\end{subfigure}
		\hfill
		\begin{subfigure}{0.4\textwidth}
			\centering
			\includegraphics[scale=0.7,page=32]{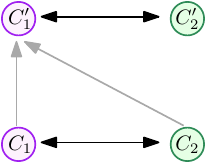}
			\caption{Conflict graph of multiple clause gadgets with positive literals extending to different directions.}
			\label{fig:cconflictpos1}
		\end{subfigure}
		\begin{subfigure}{0.4\textwidth}
			\centering
			\includegraphics[scale=0.4,page=33]{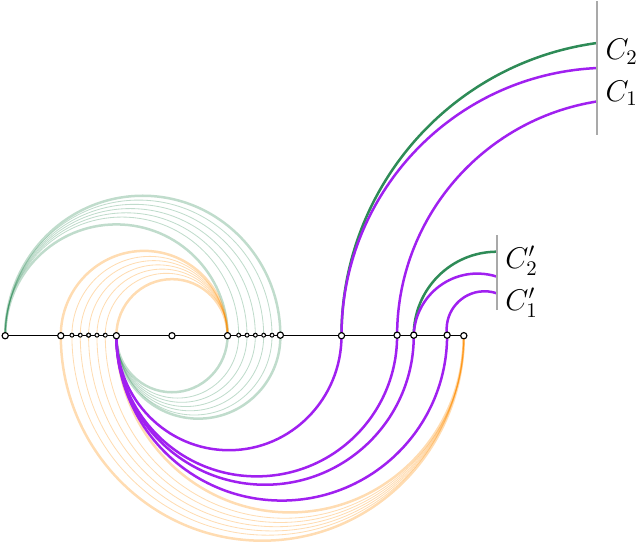}
			\caption{Multiple positive clause gadgets extending to the right.}
			\label{fig:clausepos2}
		\end{subfigure}
		\hfill
		\begin{subfigure}{0.4\textwidth}
			\centering
			\includegraphics[scale=0.7,page=34]{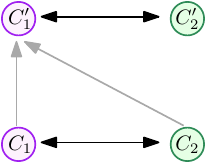}
			\caption{Conflict graph of multiple clause gadgets with positive literals extending to the right.}
			\label{fig:cconflictpos2}
		\end{subfigure}
		\begin{subfigure}{0.4\textwidth}
			\centering
			\includegraphics[scale=0.4,page=35]{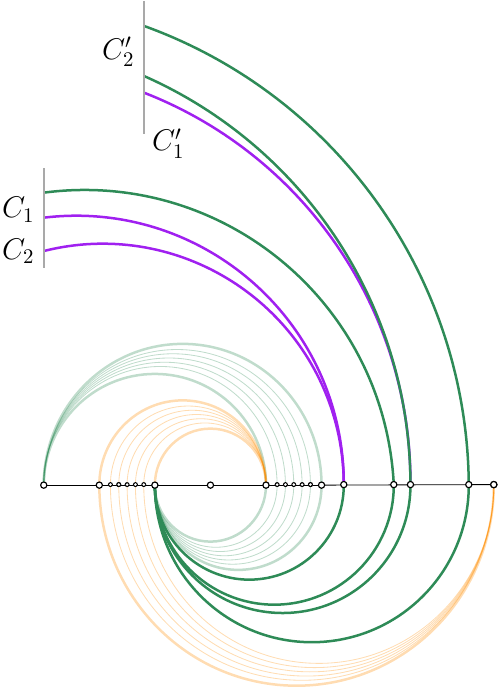}
			\caption{Multiple positive clause gadgets extending to the left.}
			\label{fig:clausepos3}
		\end{subfigure}
		\hfill
		\begin{subfigure}{0.4\textwidth}
			\centering
			\includegraphics[scale=0.7,page=36]{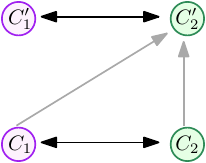}
			\caption{Conflict graph of multiple clause gadgets with positive literals extending to the left.}
			\label{fig:cconflictpos3}
		\end{subfigure}
		\caption{Multiple positive clause gadgets and their conflict graph.}
		\label{fig:conflictpos}
	\end{figure}
	
	The idea for the positive clause gadget is the following: If the triangle pairs corresponding to $x_i$ are in $S$, then no triangles corresponding to $\bar{x}_i$ are in $S$. We can add $C_1$ without having a bidirected conflict with the variable gadgets. If $C_1$ is in $S$, then this corresponds to the clause $x_i\lor x_j$ being satisfied because $x_i$ holds. Conversely, if triangle pairs of the form $\bar{x}_i$ are in $S$, we cannot add $C_1$ since it would cause a bidirected conflict and the set is not acyclic. We make a similar observation for having triangle pairs for $x_j$ and $C_2$ in $S$. If $C_2$ is in $S$, then this corresponds to the clause $x_i\lor x_j$ being satisfied by $x_j$ being true. It might be the case that both, triangle pairs from $x_i$ and triangle pairs from $x_j$, are in $S$. Then we could add either $C_1$ or $C_2$ to $S$. We cannot add both at the same time since the two have a bidirected conflict between them. As a convention, we will always add the above pair $C_2$. This concludes the high level idea of the positive clause gadget. For every clause gadget, we add no pair of triangles to our set $S$ if the clause is not fulfilled by the truth assignment. Conversely, we add exactly one pair of triangles, if the clause is satisfied. If we combine this with the way in which we add triangle pairs from variable gadgets, we observe that $S$ does not induce any bidirected conflicts.
	
	The remainder of the discussion of positive clause gadgets deals with single-directed conflicts. We analyse all the single-directed conflicts that can occur between positive clause gadgets and variable gadgets as well as between different positive clause gadgets. The goal is to argue that all single-directed conflicts that occur in the set $S$ that we construct along the way fit into the description of \cref{lem:acyclic}. This will guarantee in the end that $S$ induces an acyclic subgraph of~$H$.
	
	\Cref{fig:conflictpos} on the left side depicts all ways to combine multiple clause gadgets at the same variable gadget. By doing so, also conflicts between triangle pairs of different clause gadgets arise. In \cref{fig:conflictpos} on the right side, we display all directed conflicts to potential other clause gadgets that might be incident at the gadget of $x_i$ or $x_j$. We collect all the conflicts in the following lemma. Recall that we have a crossing whenever vertices of edges in $\Lambda(\Delta_i)$ and $\Lambda(\Delta'_j)$ appear in an alternating fashion along the boundary. For all combinations of positive clause gadgets, we study, which edges cross.

	\begin{lemma}\label{obs:conflict}
		All conflicts that occur in positive clause gadgets together with their variable gadgets are of the following forms:
		\begin{itemize}
						\setlength\itemsep{0em}
			\item Bidirected conflicts of the form $x_i\leftrightarrow\bar{x}_i$, $\bar{x}_j \leftrightarrow C_1$, $\bar{x}_i\leftrightarrow C_2$, $C_1\leftrightarrow C_2$. 
			\item Directed conflicts between triangle pairs of the same type ($x_j \to C_2$, $C_2 \to C'_2$, $C_1\to C'_1$).
			\item Directed conflicts from above pairs to crossing pairs ($x_i\to C_1$, $C_2 \to C'_1$).
			\item Directed conflicts from above pairs to below pairs ($C_2\to \bar{x}_i$).
			\item Directed conflicts from crossing pairs to below pairs ($C_1\to\bar{x}_j$).
			\item Directed conflicts from crossing pairs $C_1$ to above pairs $C_2$ at the same variable.
		\end{itemize}
	\end{lemma}
	
	For the last point of \cref{obs:conflict} we observe that we have a directed conflict from a crossing pair to an above pair. We have to avoid conflicts of this form to be able to guarantee acyclicity of $S$ if we want to apply \cref{lem:acyclic}. To overcome this, our convention on tiebreaks between $C_1$ and $C_2$ comes into action. Whenever we could put both $C_1$ and $C_2$ from the same positive clause gadget into $S$, we choose $C_2$, by our convention. That way, if we put one $C_2$ above pair of a clause in $S$, then we will never put some $C'_1$ crossing pair from a different clause that contains the same variable in our constructed subset. This follows from a simple argument: Since $C_2$ is in $S$, so is $x_i$ and not $\bar{x}_i$, therefore we could have added $C'_2$ and would also have done so by our convention.
	
	\vspace{-20pt}
	
	\begin{figure}[htb]
		\centering
		\begin{subfigure}{0.4\textwidth}
			\centering
			\includegraphics[scale=0.4,page=13]{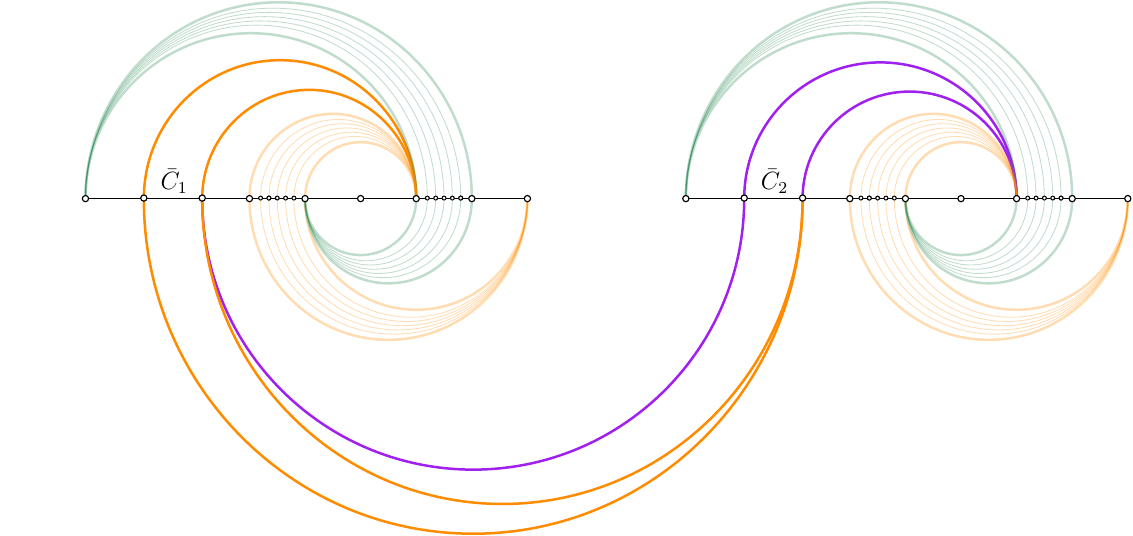}
			\caption{The clause gadget for a negative clause with the relevant variable gadgets.}
			\label{fig:clauseneg}
		\end{subfigure}
		\hfill
		\begin{subfigure}{0.4\textwidth}
			\centering
			\includegraphics[scale=0.7,page=15]{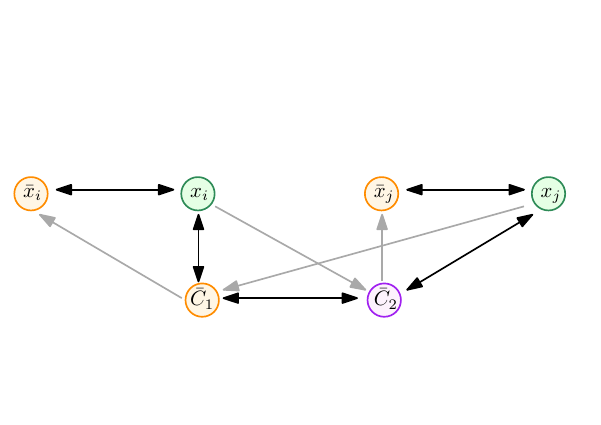}
			\caption{Conflict graph of a clause gadget with negative literals.}
			\label{fig:cconflictneg}
		\end{subfigure}
		\caption{The negative clause gadget and its conflict graph.}
	\end{figure}
	
	\subparagraph{Negative clause gadgets.} Consider a clause from the lower half plane that contains two negative literals $\bar{x}_i\lor\bar{x}_j$ with $i<j$. We add a clause gadget that consists of two pairs of triangles $\bar{C}_1$ and $\bar{C_2}$ as depicted in \cref{fig:clauseneg}. As before, we use a single vertex to represent all triangle pairs arising from $x_i$ and $\bar{x}_i$ in the conflict graph in \cref{fig:cconflictneg}. All conflicts that are depicted hold for every pair of triangles that result from subdividing the original triangles.
	
	We observe that the resulting conflict graph has three bidirected conflicts between $x_i$ and~$\bar{C}_1$, between $x_j$ and $\bar{C}_2$, and between $\bar{C_1}$ and $\bar{C}_2$. The high level idea for the negative clause gadget is that, if for one variable gadget we choose the negative assignment of the variable, in the process fulfilling the clause, then we can add one of the clause pairs to $S$. If $S$ contains $\bar{x}_i$, then we add~$\bar{C}_1$. If $S$ contains $\bar{x}_j$, then we add $\bar{C}_2$. If $S$ contains both $\bar{x}_i$ and $\bar{x}_j$, we can still only add one of $\bar{C}_1$ and $\bar{C}_2$. We set the convention, that if we can choose between the two, we add $\bar{C}_1$ to~$S$. We again conclude from the high level idea that by adding pairs of triangles from negative clause gadgets this way, we never obtain bidirected conflicts in $S$.
	
	\begin{figure}[t]
		\centering
		\begin{subfigure}{0.4\textwidth}
			\centering
			\includegraphics[scale=0.4,page=37]{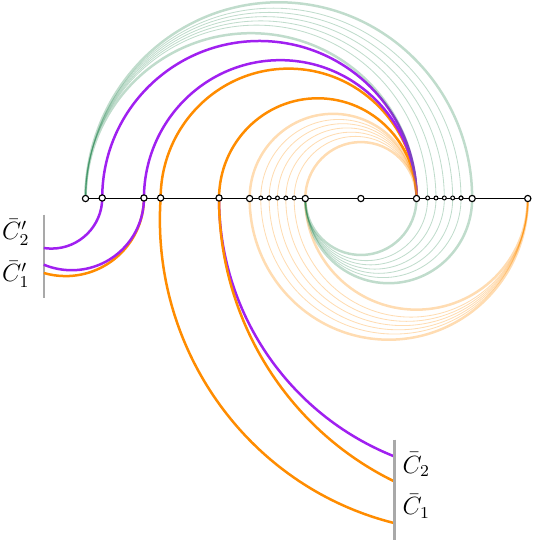}
			\caption{Multiple negative clause gadgets extending to different sides.}
			\label{fig:clausenegm}
		\end{subfigure}
		\hfill
		\begin{subfigure}{0.4\textwidth}
			\centering
			\includegraphics[scale=0.7,page=38]{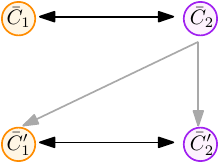}
			\caption{Conflict graph of multiple clause gadgets with negative literals extending to different sides.}
			\label{fig:cconflictnegm}
		\end{subfigure}
		\centering
		\begin{subfigure}{0.4\textwidth}
			\centering
			\includegraphics[scale=0.4,page=39]{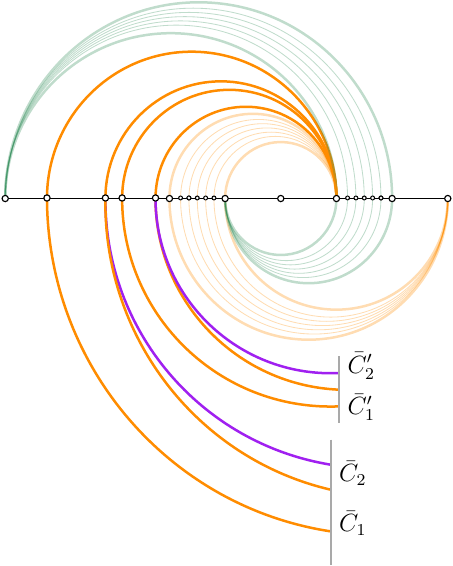}
			\caption{Multiple negative clause gadgets extending to the right.}
			\label{fig:clausenegm2}
		\end{subfigure}
		\hfill
		\begin{subfigure}{0.4\textwidth}
			\centering
			\includegraphics[scale=0.7,page=40]{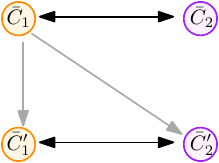}
			\caption{Multiple clause gadget gadgets with negative literals extending to the right.}
			\label{fig:cconflictnegm2}
		\end{subfigure}
			\centering
		\begin{subfigure}{0.4\textwidth}
			\centering
			\includegraphics[scale=0.4,page=41]{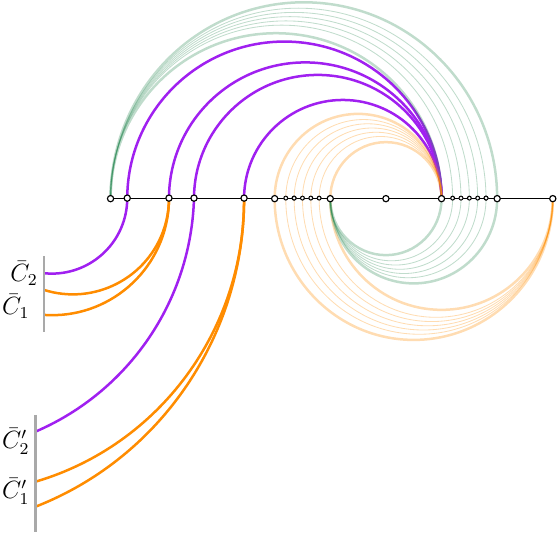}
			\caption{The clause gadget for multiple negative clauses extending to the left.}
			\label{fig:clausenegm3}
		\end{subfigure}
		\hfill
		\begin{subfigure}{0.4\textwidth}
			\centering
			\includegraphics[scale=0.7,page=42]{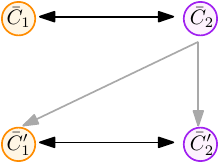}
			\caption{Conflict graph of multiple clause gadgets with negative literals extending to the left.}
			\label{fig:cconflictnegm3}
		\end{subfigure}
		\caption{Multiple negative clause gadgets and their conflict graph.}
		\label{fig:conflictneg}
	\end{figure}
	
	Again, the remaining part of the discussion of the gadget is entirely devoted to making sure that the directed conflicts are in accordance with the conditions of~\cref{lem:acyclic}.
	
	\Cref{fig:conflictneg} on the left side depicts all ways to combine multiple negative clause gadgets at the same variable gadget. When adding multiple clauses at the same gadget, new conflicts between clause gadgets arise as well. In \cref{fig:conflictneg} on the right side we depict all directed conflicts to potential other clause gadgets that might arise when having multiple clause gadgets at $x_i$ and~$x_j$. We collect all the appearing conflicts in the following lemma.
	
	\begin{lemma}\label{obs:conflict2}
		All conflicts that occur in negative clause gadgets together with their variables gadgets are of the following forms:
		\begin{itemize}
						\setlength\itemsep{0.00em}
			\item Bidirected conflicts of the form $x_i\leftrightarrow\bar{x}_i$, $x_i\leftrightarrow\bar{C}_1$, $x_j\leftrightarrow\bar{C}_2$.
			\item Directed conflicts between pairs of the same type ($\bar{C}_1 \to \bar{x}_i$, $\bar{C}_2 \to \bar{C}'_2$, $\bar{C}_1 \to \bar{C}'_1$).
			\item Directed conflicts from above pairs to crossing pairs ($x_j\to \bar{C}_1$).
			\item Directed conflicts from above pairs to below pairs ($x_j \to \bar{C}_1$).
			\item Directed conflicts from crossing pairs to below pairs ($\bar{C}_2 \to \bar{C}'_1$).
			\item Directed conflicts from below pairs $\bar{C}_1$ to crossing pairs $\bar{C}_2$ at the same variable
		\end{itemize}
	\end{lemma}
	
	As for the positive variable gadget, we want to avoid directed conflicts from the last point of \cref{obs:conflict2}, as they are not covered by \cref{lem:acyclic}. Again, we handle this by our convention on breaking ties between $\bar{C}_1$ and $\bar{C}_2$. Whenever we could put $\bar{C}_1$ and $\bar{C}_2$ from the same negative clause gadget into $S$, we choose $\bar{C}_1$. That way, if we have one $\bar{C}_1$ below pair of a clause in $S$, then we will never put some $\bar{C}'_2$ crossing pair from a different clause, but at the same variable, in $S$. Since $\bar{C}_1$ is in $S$, so is $\bar{x}_i$, therefore, we could add $\bar{C}'_1$ and would do so by our convention.
	
	\subparagraph{Combining positive and negative clause gadgets.} Finally, we consider the conflicts that can occur between different clause gadgets, where one is positive and the other is negative. Again, the objective is to make sure that our constructed set $S$ fulfills the second condition of \cref{lem:acyclic}.
	
	The easy case occurs whenever the two variables that are contained in one clause are distinct from the two variables that are contained in the second clause. We observe that every pair of triangles in a clause gadget has one \emph{short} triangle that has all its vertices within the region of one variable gadget. The other \emph{long} triangle determines conflicts. For a conflict to exist, the long triangles of the pairs have to cross. In the positive clauses the long triangles are in the initial triangulation. In the negative clauses the long triangles are in the target triangulation. For an illustration see \cref{fig:uddisjoint}. Therefore, all conflicts go from triangle pairs in positive clause gadgets to triangle pairs in negative clause gadgets. We observe the following conflicts:
	
	\begin{lemma} \label{obs:conflict3}
		Between pairs of triangles in positive clause gadgets and negative clause gadgets, where the clauses contain distinct variables, the occurring conflicts are:
		\begin{itemize}
						\setlength\itemsep{0em}
			\item Directed conflicts between crossing pairs.
			\item Directed conflicts from above pairs to crossing pairs.
			\item Directed conflicts from above pairs to below pairs.
			\item Directed conflicts from crossing pairs to below pairs.
		\end{itemize}
	\end{lemma}
	
	\begin{figure}[htb]
		\centering
		\begin{subfigure}{0.3\textwidth}
			\centering
			\includegraphics[scale=0.4,page=52]{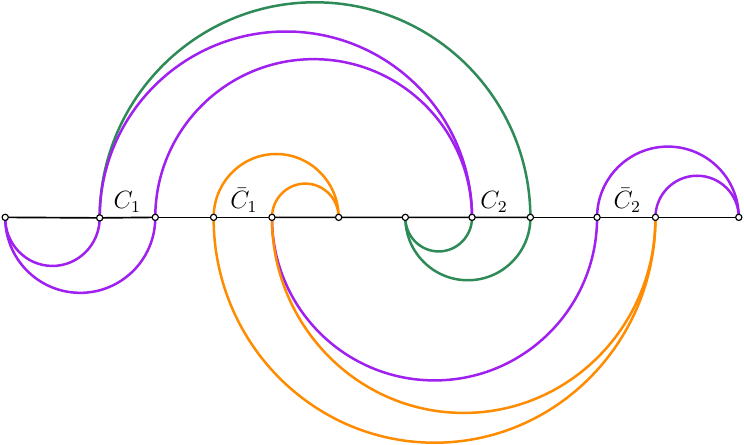}
			\caption{A positive and a negative clause that do not have a variable in common.}
			\label{fig:uddisjoint}
		\end{subfigure}
		\hfill
		\begin{subfigure}{0.3\textwidth}
			\centering
			\includegraphics[scale=0.4,page=17]{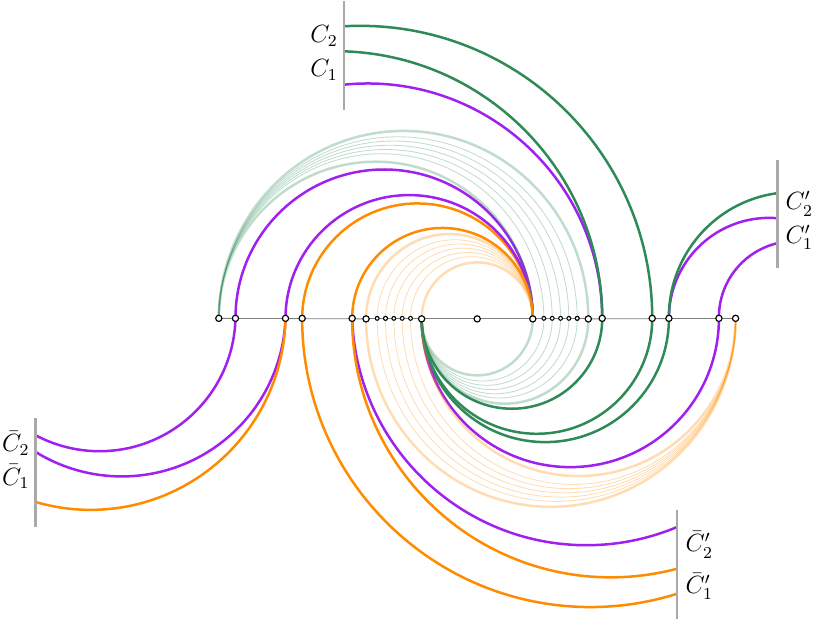}
			\caption{A variable gadget with both positive and negative clauses incident to it.}
			\label{fig:clauses}
		\end{subfigure}
		\hfill
		\begin{subfigure}{0.3\textwidth}
			\centering
			\includegraphics[scale=0.55,page=18]{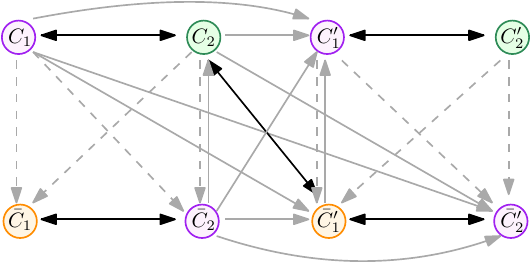}
			\caption{Conflict graph of \ref{fig:clauses}.}
			\label{fig:cconflicts}
		\end{subfigure}
		\caption{Mixing positive and negative clauses: (a) With disjoint variables; (b) and (c) at the same variable.}
	\end{figure}
	
	We observe that all the conflicts in \cref{obs:conflict3} are in accordance to the second condition of \cref{lem:acyclic}.
	
	Finally, for conflicts between positive and negative clauses that share a variable, all options are covered in \cref{fig:clauses}. The resulting conflict graph is depicted in \cref{fig:cconflicts}. Some conflicts are depicted via dashed arrows, this indicates that the conflict might or might not occur depending on how the clauses continue outside the drawing.
	
	\begin{observation}\label{obs:conflicts4}
		The conflicts that occur when mixing positive and negative clauses at a variable are
		\begin{itemize}
						\setlength\itemsep{0em}
			\item Bidirected conflicts between pairs of the form $C'_1$ and $\bar{C}'_1$, $C_2$ and $\bar{C}_2$, and $C_2$ and $\bar{C}'_1$.
			\item (Potential) directed conflicts from above pairs to crossing pairs ($C'_2 \to \bar{C}'_2$, $C_2\to\bar{C}_2$).
			\item (Potential) directed conflicts from above pairs to below pairs ($C_2\to\bar{C}_1$, $C'_2\to\bar{C}'_1$).
			\item Directed conflicts from crossing pairs to below pairs ($C_1\to \bar{C}'_1$, $C_1\to\bar{C}_1$, $C'_1\to\bar{C}'_1$).
			\item Directed conflicts between crossing pairs.
			\item Directed conflicts from the below pair $\bar{C}'_1$ to the crossing pair $C'_1$ and from the crossing pair $\bar{C}_2$ to the above pair $C_2$.
		\end{itemize}
	\end{observation}
	
	Conflicts of the first and last point occur between a triangle pair that is only added to $S$ if the variable has assignment \textsc{True} and another triangle pair that requires that the same variable is set to \textsc{False}. Therefore, these conflicts will not occur between triangle pairs in $S$. All other conflicts are allowed by the second condition of \cref{lem:acyclic}.
	
	\subparagraph{The reduction.} We combine all gadgets to one instance. Consider a given instance $\phi$ of \textsc{Planar Separable Monotone Max-2SAT} on $\lvert X \rvert$ variables and $\lvert C \rvert$ clauses with an aligned drawing $D(G_\phi)$ of the clause variable incidence graph $G_\phi$ and an integer $k'$. We construct two triangulations $T$ and $T'$ with their linear representation according to $D(G_\phi)$. First, we place variable gadgets in the linear representation. The gadgets from left to right represent the variables in the same order as they appear on the horizontal separating line in $D(G_\phi)$. Further, we place the variable gadgets such that the rightmost vertex of one variable gadget coincides with the leftmost vertex of the variable gadget immediately to its right.
	
	\begin{figure}[htb]
		\centering
		\begin{subfigure}{0.45\textwidth}
			\centering
			\includegraphics[scale=0.9,page=19]{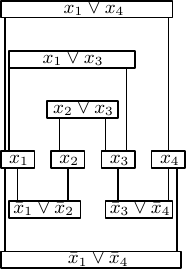}
			\label{fig:max2sat}
		\end{subfigure}
		\begin{subfigure}{0.45\textwidth}
			\centering
			\includegraphics[scale=0.39,page=20]{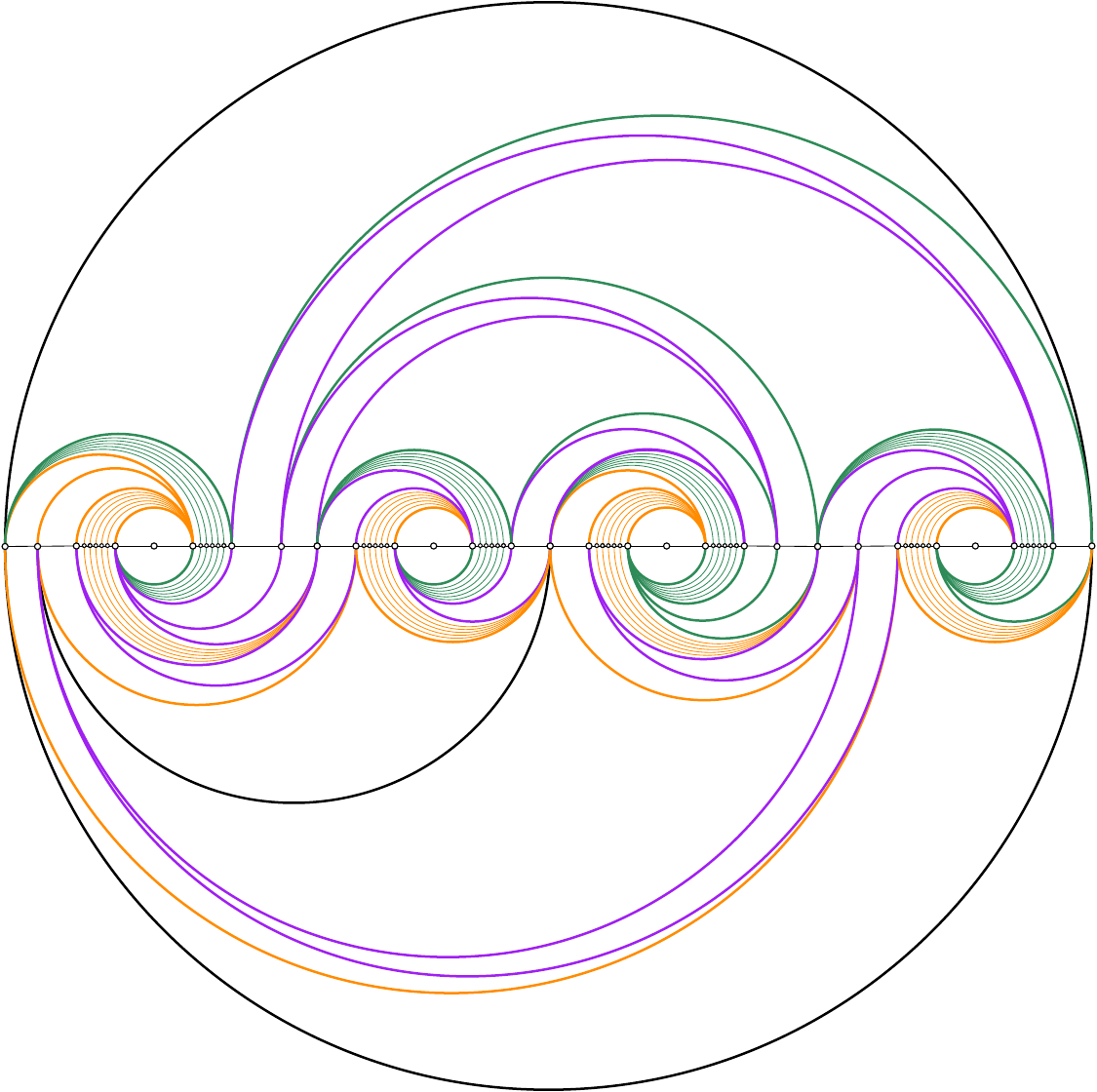}%
			\label{fig:reduction1}
		\end{subfigure}
		\hfill
		\caption{A planar \textsc{2SAT} formula in an aligned drawing (left) and two triangulations $T$ and~$T'$ that result from the reduction (right).}
		\label{fig:reduction}
	\end{figure}
	
	Next, we add a clause gadget for every clause. The drawing of the clause incidence graph $D(G_\phi)$ determines the positioning of the gadgets. We say that a clause $C_1$ \emph{covers} another clause $C_2$ if in $D(G_\phi)$ $C_2$ is contained in the bounded face whose boundary is formed by the horizontal spine containing the variables, the rectangle representing $C_1$ and the two clause-variable-incidence edges at $C_1$. For a positive clause $C_2$ that is covered by a positive clause $C_1$, we can arrange the positive gadgets to resemble this drawing. We do so by placing the clause gadget of $C_2$ in the face that is bounded by the spine vertices from below and the above and crossing pair pair of $C_1$ from above. In a symmetric way, we can describe the placement of negative clause gadgets.
	
	Further, we add the clauses such that we have no ``gaps" on the spine between the constructed triangles. The left vertex of the pairs coming from $x_i$ is the same vertex as the left vertex of the first triangle of a negative clause gadget. The right vertex of the triangle of one negative clause gadget is the left vertex of the negative clause gadget that comes immediately after in the linear representation. The right vertex of the last negative clause gadget coincides with the left vertex of the triangles coming from $\bar{x}_i$. We repeat this for the second half of the gadget for the positive clause gadgets in a similar manner. That way, all edges of the spine are part of some triangle and $\Gamma$ is fully determined by the variable gadgets and clause gadgets.
	
	Since, $G_\phi$ is planar, we end up with two plane graphs, one above the spine of the linear representation and one below. We arbitrarily extend both of them to triangulations. For an illustration, see \cref{fig:reduction}.
	
	\begin{proposition} 
		Let $\phi$ be a planar separable monotone 2\textsc{SAT} formula with $\lvert X \lvert$ variables and $\lvert C \rvert$ clauses, together with a clause-variable incidence graph $G_\phi$ and an aligned drawing $D(G_\phi)$. Further, let $T$ and $T'$ be the triangulations that are obtained from $D(G_\phi)$ by applying the described reduction. For a given integer $k'$ the following statements are equivalent:
		\begin{itemize}
			\item $\phi$ has a variable assignment that fulfills at least $k'$ clauses.
			\item The conflict graph $H$ of $T$ and $T'$ with their linear representation has an acyclic subset $S$ of size $\lvert X \rvert (\lvert C \rvert+1)+k'$.
		\end{itemize}
	\end{proposition}
	
	\begin{proof}
		Assume $\phi$ has a variable assignment that fulfills at least $k'$ clauses. We construct an acyclic subset $S$ of the conflict graph of the prescribed size. For every variable $x_i$ we add all of the triangles labeled $x_i$ to $S$ whenever $x_i$ has truth assignment \textsc{True} and all triangles labeled~$\bar{x}_i$ otherwise. For every fulfilled clause we add $C_1$ or $C_2$, respectively $\bar{C}_1$ or $\bar{C}_2$ to $S$, depending on which of the two incident variables is satisfied. Recall that we break ties in favor of $C_2$ and~$\bar{C}_1$. The resulting set $S$ contains $\lvert C \rvert+1$ pairs for each of the $\lvert X \rvert$ variables and at least $k'$ pairs from the fulfilled clauses. This gives us a total of $\lvert X \rvert(\lvert C \rvert+1)+k'$ pairs. By Lemma \ref{obs:conflict}, \ref{obs:conflict2}, \ref{obs:conflict3}, and~\ref{obs:conflicts4} in addition to \cref{lem:acyclic}, $S$ is acyclic. The ordering that witnesses the acyclicity takes first all above pairs, then all crossing pairs then all below pairs in $S$. Within pairs of the same type we sort the pairs by their edge on the spine from left to right. 
		
		Now assume the conflict graph has an acyclic subset $S$ of size at least $\lvert X \rvert(\lvert C \rvert+1)+k'$. Then it has to contain at least one pair of triangles from each variable gadget or otherwise, we miss out on $\lvert C \rvert+1$ triangle pairs, which makes it impossible to obtain the desired cardinality of $S$. To each variable $x_i$ assign \textsc{True} if one of the pairs of triangles labeled $x_i$ is in $S$ and \textsc{False} otherwise. This truth assignment leads to at least~$k'$ clauses being satisfied as witnessed by the at least~$k'$ pairs of triangles that have to come from clause gadgets and that encode truth assignments that lead to the clause being fulfilled.
	\end{proof}
	
	This completes the proof of \cref{prop:achard}.

	\section{Conflict Graph based Upper Bounds on the Flip Distance}\label{sec:flipsequence} 
	
	Throughout this section, we again consider triangulations $T$ and $T'$ and their resulting $\beta$-blow-ups $\beta T$ and $\beta T'$. $H=H(T,T')$ describes the conflict graph (with vertex set $\Gamma$) of $T$ and $T'$, and $ac(H)$ the size of a maximum acyclic set $S$ of $H$. We describe how to construct a flip sequence from $\beta T$ to $\beta T'$ that witnesses the proof of \cref{prop:upper}.
	
	Let $R = \Gamma \setminus S$. Elements of $S$ will be dealt with in the \emph{direct phase}, and elements of $R$ will be dealt with in the \emph{indirect phase}.
	
	Roughly speaking, for each pair in $R$, the flip sequence uses two flips per fan edge. For each pair in $S$ the flip sequence uses only one flip per fan edge, after some relatively small number of ``set-up'' flips.
	
	The indirect phase consists of a \emph{beginning} sequence of $\beta\lvert R \rvert$ flips that transform $T$ to $T_R$, and an \emph{ending} sequence of $\beta \lvert R\rvert$ flips, which, in reverse, transforms $T'$ to $T'_R$. For each pair $(\Delta,\Delta')$ in $R$ we perform one flip on each edge of the fan $\Lambda(\Delta)$ in the beginning sequence, and one flip on each edge of the fan $\Lambda(\Delta')$ in the ending sequence. For one example pair, we illustrate these flips and their order in \cref{fig:indirect}. Different pairs in $R$ are handled independently, and their relative order does not matter. The edges resulting from the indirect phase, denoted by $\Lambda(\Delta\Delta')$ connect vertices from $V(\Delta)$ and in particular do not cross any edges from other triangle pairs. This can be interpreted as removing $(\Delta,\Delta')$ and all its incident conflicts from $H$.
	
	The direct phase transforms $T_R$ into $T'_R$. It operates on pairs in $S$ in an ordering $(\Delta_i,\Delta'_i)$ for $i=1,\ldots,ac(H)$, where no conflict edge goes from a higher index pair to a lower index pair. For each pair $(\Delta_i,\Delta'_i)\in S$, we first perform up to $2n-6$ ``set-up'' flips, and are then able to flip each edge of $\Lambda(\Delta_i)$ directly to an edge of $\Lambda(\Delta'_i)$, as shown in~\cref{fig:direct}. This is where the work is done to create short flip sequences. It is known that flipping every edge twice is sufficient to go from one triangulation to another~\cite[Lemma 2]{sleator1986rotation}. Spending only one flip per edge on a large set of edges is where we save flips.
	
	A high level overview of the two phases is depicted in \cref{fig:schema}.
	
	\begin{figure}[t]
		\centering
		\includegraphics[scale=0.8,page=55]{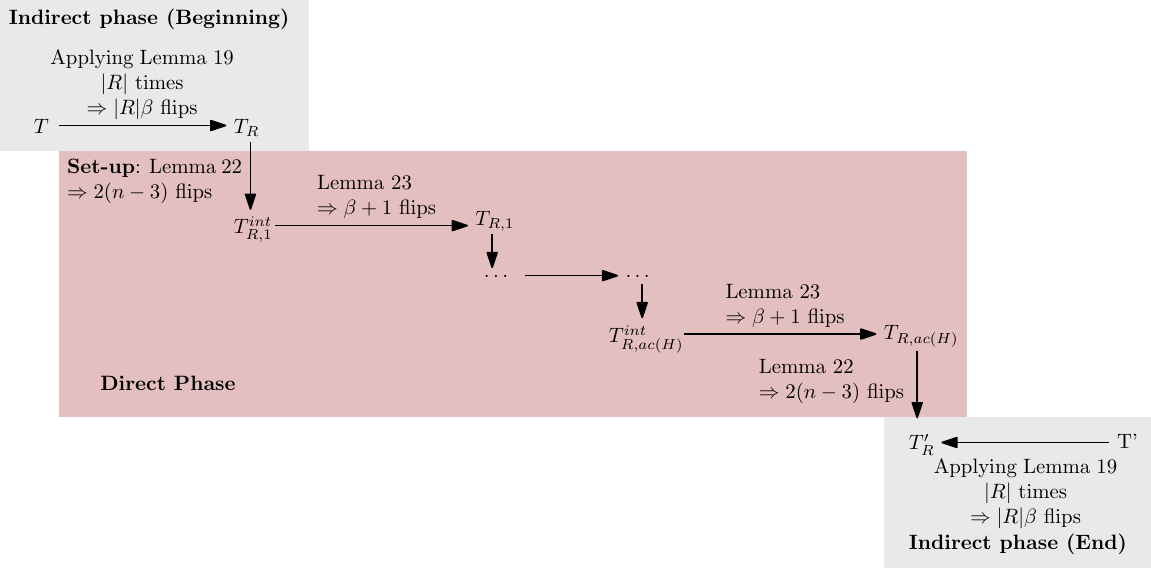}
		\caption{Schematic high level overview of the flip sequence.}
		\label{fig:schema}
	\end{figure}
	
	We will now explain the two phases in further detail.
	
	\subparagraph{Indirect Phase.} We perform the following flip sequence for every triangle pair $(\Delta,\Delta')\in R$. For this, we denote the vertices of $\Delta$ by $v_{1}$, $v_{2}$, and $v_{3}$ such that $v_{2}$ is the vertex that is not incident to the spine edge of $\Delta$. Further, denote the vertices of $\Delta'$ by $v_{1}$, $v'_{2}$, and $v_{3}$.
	
	Let $P$ be the polygon that is bounded by $V^O(\Delta)$, $v_{1}$, $v_{2}$, and $v_{3}$. Similarly, $P'$ is the polygon that is bounded by $V^O(\Delta')$, $v_{1}$, $v'_{2}$, and $v_{3}$. 
	
	The part of the beginning (resp. end) phase that is relevant for $(\Delta,\Delta')$ happens entirely in $P$ (resp. $P'$). Inside $P$ and $P'$, we flip to triangulations $T_P$ and $T_{P'}$ that have all vertices incident to $v_1$. In particular, all interior edges of $T_P$ and $T_{P'}$ will coincide as they connect $v_1$ to the same set of vertices. As mentioned, we denote the collection of these edges as $\Lambda(\Delta\Delta')$.
	
	We now analyse, how many flips it takes to build $T_P$ and $T_{P'}$.
	
	\begin{lemma}\label{lem:weg}
		Any triangulation of $P$ (resp. $P'$) can be flipped to $T_P$ (resp. $T_{P'}$) in at most $\beta$ flips.
	\end{lemma}
	
	\begin{proof}
		As for example argued in~\cite[Lemma 2]{sleator1986rotation} it is always possible to increase the number of edges that are incident to a given, fixed vertex by $1$ in every flip. Therefore, it takes at most $\beta$ flips to flip triangulations $T_P$ and $T_{P'}$ that have all $\beta$ interior edges incident to $v_1$.
	\end{proof}
	
	\begin{figure}[htb]
		\centering
		\includegraphics[scale=0.8,page=21]{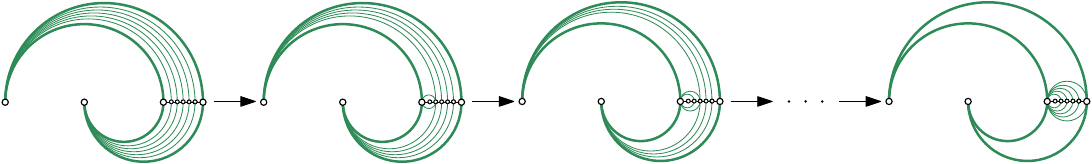}
		\caption{The indirect phase executed on an above pair.}
		\label{fig:indirect}
	\end{figure}
	
	In one iteration of the direct phase, we spend $2\beta$ flips to obtain triangulations that have $\beta$ edges more in common than their predecessors. The interiors and edge sets of the involved polygons in $T$ are disjoint. The same holds for $T'$. Therefore, all the triangle pairs in $R$ can be handled independently in any order. In total, we spent $2\beta(\lvert \Gamma \rvert-ac(H))$ flips to obtain $T_R$ and $T'_{R}$. Furthermore, all edges of pairs in $R$ now coincide between $T_R$ and $T'_{R}$. As such, these edges do not cross any fan edges of any of the pairs in $S$. Since we did not introduce any new triangle fans that may cross existing triangle fans, the remaining conflict graph has as vertices only the triangle pairs of $S$. We collect our findings from the analysis of the indirect phase in the following lemma.
	
	\begin{lemma}\label{lem:indirect}
		The shortest flip sequence from $T$ to $T'$ is at most $2\beta(\lvert \Gamma \rvert - ac(H))$ flips longer than the shortest flip sequence from $T_{R}$ to $T'_{R}$. The conflict graph $H(T_R,T'_{R})$ has as vertex set exactly the triangle pairs in $S$ and is as such acyclic.
	\end{lemma}
	
	\subparagraph{Direct Phase.} We start working in the triangulation $T_R$ and want to transform it to~$T'_R$. For this purpose $T_{R,i}$ denotes a triangulation that coincides with $T'_R$ in all fans $\Lambda(\Delta_	j)$ for the first $i$ triangle pairs $(\Delta_j,\Delta'_j)\in S$. All other pairs of triangles $(\Delta_j,\Delta'_j)\in S$ with $j>i$ the triangulation $T_{R,i}$ still coincides with $\beta T$, by containing the edges of $\Lambda(\Delta_j)\cup \Delta_j$. Further, throughout the entire direct phase, we never flip edges in triangle pairs in $R$ again. All triangulations in this phase coincide with $T_R$ and $T'_R$ in all pairs $(\Delta,\Delta') \in R$ by containing the edges in $\Lambda(\Delta\Delta')$. The described sets of edges does not fix $T_{R,i}$ entirely. There might $n-3$ edges that are not determined by this description (but it will be clear where they come from in the description of the flip sequence). We denote the set of non-specified edges by $A_{R,i}$.

	We iteratively flip triangulations $T_{R,i-1}$ to triangulations $T_{R,i}$. The flip sequence flips $\beta +1$ edges in $\Lambda(\Delta_i)\cup \Delta_i$ to $\beta+1$ edges of $\Lambda(\Delta'_i)\cup\Delta'_i$ using a single flip per edge. To enable these direct flips, we first need a set-up that does not require too many flips. In our case ``not to many" means only depending on $n$, but not $\beta$. To ensure this, the set-up phase only flips edges among the~$n-3$ interior edges of $T_{R,i-1}$ that are not part of any fan of a triangle pair.
	
	We proceed to describe the details of one set-up of the direct phase. Assume we are given the triangulation $T_{R,i-1}$ that contains all the fan edges as described above. We pick the pair $(\Delta_i,\Delta'_i)$ from $S$, which is a source of $H[S\setminus\{(\Delta_1,\Delta'_1),\ldots,(\Delta_{i-1},\Delta'_{i-1})\}]$. We again denote the vertices of $\Delta_i$ by $v_{i,1}$, $v_{i,2}$, and $v_{i,3}$ such that $v_{i,2}$ is the vertex that is not incident to the edge of $\Delta_i$ on the boundary. We further denote the vertices of $\Delta'_i$ by $v_{i,1}$, $v'_{i,2}$, and $v_{i,3}$.
	
	We construct an \emph{intermediate} triangulation $T_{R,i}^{int}$ by first removing the edges in $A_{r,i-1}$ from~$T_{R,i-1}$. Next, we add edges $v_{i,1}v'_{i,2}$ and $v_{i,2}v'_{i,2}$ to it and call these two edges the \emph{supporting} edges. We verify that the resulting collection of edges is still plane. From the fact that $S$ is an acyclic subset of the conflict graph it follows immediately that for any $i$ the collection $(\Delta'_1\cup \Lambda(\Delta'_1))\cup \ldots \cup (\Delta'_{i-1}\cup \Lambda(\Delta'_{i-1}))\cup (\Delta_{i}\cup \Lambda(\Delta_{i}))\cup \ldots \cup (\Delta_{ac(H)}\cup \Lambda(\Delta_{ac(H)}))$ is crossing free. We only need to check that adding edges $v_{i,1}v'_{i,2}$ and $v_{i,2}v'_{i,2}$ still keeps the graph plane, which we do in \cref{lem:support}.
	
	\begin{figure}[htb]
		\centering
		\begin{subfigure}{0.3\textwidth}
			\centering
			\includegraphics[scale=0.7,page=22]{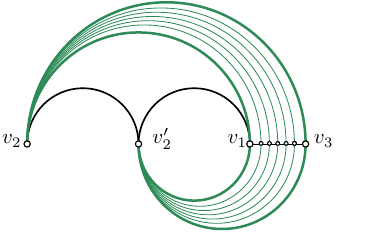}
			\caption{An above pair and its supporting edges}
		\end{subfigure}
		\hfill
		\begin{subfigure}{0.3\textwidth}
			\centering
			\includegraphics[scale=0.7,page=25]{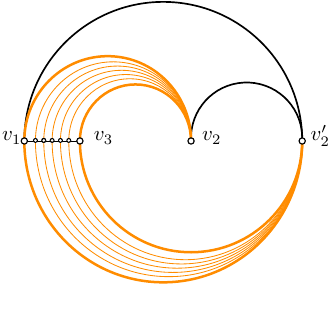}
			\caption{A below pair and its supporting edges}
		\end{subfigure}
		\hfill
		\begin{subfigure}{0.3\textwidth}
			\centering
			\includegraphics[scale=0.7,page=28]{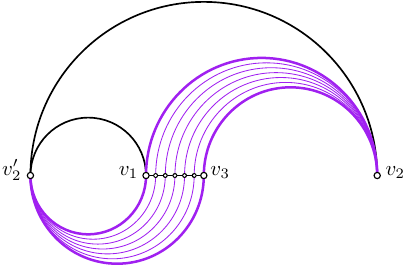}
			\caption{A crossing pair and its supporting edges}
		\end{subfigure}
		\caption{Supporting edges for direct flips for three of the triangle pairs}
	\end{figure}
	
	\begin{lemma}\label{lem:support}
		In the current step of the construction $v_{i,1}v'_{i,2}$ and $v_{i,2}v'_{i,2}$ are not crossed in $T_{R,i}^{int}$.
	\end{lemma}
	
	\begin{proof}
		We start by observing that $v_{i,1}v'_{i,2}$ and $v_{i,2}v'_{i,2}$ cannot be crossed by any edge in any $\Lambda(\bar{\Delta}_j\bar{\Delta}'_j)$ for any pair $(\bar{\Delta}_j,\bar{\Delta}'_j)\in R$. An edge that crosses an edge in $\Lambda(\bar{\Delta}_j\bar{\Delta}'_j)$ would need one vertex in $V^O(\Delta_j)$, which none of the two edges $v_{i,1}v'_{i,2}$ and $v_{i,2}v'_{i,2}$ have. So it suffices to check for crossings with edges of fans $\Lambda(\Delta_j)$ and $\Lambda(\Delta'_j)$ with $(\Delta_j,\Delta'_j)\in S$. We verify the statement for above, below, and crossing pairs. Note that the three pairs also have a mirrored counterpart. For those the statement follows via a symmetric argument by adapting the order of the vertices.
		
		First, assume $(\Delta_i,\Delta'_i)$ is an above pair. An edge $f$ that crosses $v_{i,1}v'_{i,2}$ or $v_{i,2}v'_{i,2}$, but no edge of $\Lambda(\Delta_i)$, has to have one vertex $u$ with $v_{i,2}<u<v'_{i,2}$ and one vertex $w$ with $v'_{i,2}<w<v_{i,1}<v_{i,3}$. Since all edges of $\Lambda(\Delta'_i)$ have as one vertex $v'_{i,2}$ and another vertex $v_i \in V^O(\Delta_i)$ with $v_{i,1}<v_i<v_{i,3}$, we conclude that $u<v'_{i,2}<w<v_i$ and therefore $f$ crosses any edge in $\Lambda(\Delta'_i)$. If $f$ is in $\Lambda(\Delta_j)\cup \Delta_j$ for some $(\Delta_j,\Delta'_j)\in S$ and $j>i$, then this contradicts that $(\Delta_i,\Delta'_i)$ was a source of $H[S\setminus\{(\Delta_1,\Delta'_1),\ldots,(\Delta_{i-1},\Delta'_{i-1})\}]$, since there is a directed edge $(\Delta_j,\Delta'_j)\to (\Delta_i,\Delta'_i)$. Otherwise, if $f$ is in $\Lambda(\Delta'_j)\cup\Delta'_j$ for some $(\Delta_j,\Delta'_j)\in S$ and $j<i$, then this contradicts the planarity of $\beta T'$. We conclude that neither of the two cases can occur.
		
		If $(\Delta_i,\Delta'_i)$ is a below pair, then $f$ must have one vertex $u$ with $v_{i,2}<u<v'_{i,2}$ and another vertex $w$ with $w<v_{i,1}$ or $v'_{i,2}<w$. For any edge $v_iv'_{i,2}$ in $\Lambda(\Delta'_i)$ we get that in the former case $w < v_{i,1} < v_i < v_{i,3} < v_{i,2} < u <v'_{i,2}$ and in the latter case $v_i < v_{i,3} < v_{i,2} < u < v'_{i,2} < w$. In both cases $f$ crosses all edges of $\Lambda(\Delta'_i)$, because the vertices of the edges appear in alternating order along the boundary. Again, this either contradicts $(\Delta_i,\Delta'_i)$ being a source or $\beta T'$ being planar.
		
		Finally let $(\Delta_i,\Delta'_i)$ be a crossing pair. The edge $f$ has one vertex $u$ with $v'_{i,2} < u < v_{i,1}$ and one vertex $w$ with $w<v'_{i,2}$ or $v_{i,2} < w$. In the former case $w < v'_{i,2} < u < v_{i,1} < v_i$ and in the latter case $u < v_{i,1} < v_i < v_{i,3} < v_{i,2} < w$. In both cases, edges in $\Lambda(\Delta'_i)$ are crossed by $f$ and either $(\Delta_i,\Delta'_i)$ is not a source or $\beta T'$ is not planar, a contradiction.
	\end{proof}
	
	We complete $T_{R,i}^{int}$ arbitrarily to a triangulation by adding less than $n-3$ additional edges. By construction, $T_{R,i-1}$ and $T_{R,i}^{int}$ coincide in at least $\beta\lvert\Gamma\rvert$ interior edges. Consequently, they can differ in at most $n-3$ interior edges. We can transform $T_{R,i-1}$ into $T_{R,i}^{int}$ via a flip sequence that uses at most $2(n-3)$ flips, by \cref{lem:symmdiff}, which is a very direct application of results in \cite{sleator1986rotation}, especially Lemma~2.
	
	\pagebreak
	
	\begin{lemma}\label{lem:symmdiff}
		The flip distance between two triangulations $T_1$ and $T_2$ can be bounded by $2\lvert T_1 \setminus T_2\rvert$.
	\end{lemma}
	
	\begin{proof}
		We can split the triangulations along a happy edge and obtain smaller triangulations. The flip distance for the original pair of triangulations is then at most the sum of the flip distances for the two smaller pairs of triangulations. We split along all happy edges until we are left with many small instances of the flip distance problem such that all interior edges are not happy edges. In particular, there are $\lvert T_1\setminus T_2\rvert$ interior edges over all instance combined.
		
		By \cite[Lemma 2]{sleator1986rotation} there exists a flip sequence between any pair of triangulations of a convex polygon that flips every interior edge at most twice. This yields the desired flip distance.
	\end{proof}
	
	\cref{lem:symmdiff} concludes the set up. Now, we can greedily insert one edge of $\Lambda(\Delta'_i)$ after another thanks to the existence of $v_1v'_2$ and $v_2v'_2$, as portrayed in \cref{fig:direct} and argued in \cref{lem:directphase}. \cref{lem:directphase} is again a direct result from the proof of Lemma~2 in \cite{sleator1986rotation}.
	
	We denote by $T_{R,i}=(T_{R,i}^{int}\setminus(\Lambda(\Delta_i)\cup \{v_1v_2\}))\cup(\Lambda(\Delta'_i)\cup\{v'_2v_3\})$. 
	
	\begin{figure}[htb]
		\centering
		\includegraphics[scale=0.8,page=29]{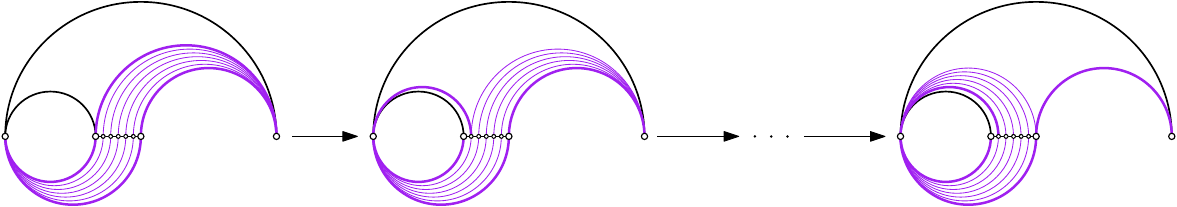}
		\caption{The direct phase executed on a crossing pair.}
		\label{fig:direct}
	\end{figure}
	
	\begin{lemma}\label{lem:directphase}
		The flip distance between $T_{R,i}^{int}$ and $T_{R,i}$ is $\beta+1$.
	\end{lemma}
	
	\begin{proof}
		Observe that $T_{R,i}^{int}$ and $T_{R,i}$ already coincide in all edges except for the edges that are inside a convex polygon $\bar{P}$ that has as vertices $v_1$, $v_2$, $v_2'$, and $v_3$ as well as all the vertices of $V^O(\Delta_i)$. As for example argued in \cite[Lemma 2]{sleator1986rotation} we can flip any triangulation of $\bar{P}$ to the triangulation of $\bar{P}$ that has all edges incident to the same vertex while flipping every interior edge at most once. In our case, we want all edges to be incident to $v'_2$. The triangulation of $\bar{P}$ with all vertices incident to $v'_2$ contains exactly all the edges of $\Lambda(\Delta'_i)$.
	\end{proof}
	
	We combine our results on the indirect phase in \cref{lem:indirect} together with the flip sequences of the direct phase in \cref{lem:symmdiff} and \cref{lem:directphase} to obtain one flip sequence that transforms $\beta T$ to $\beta T'$ (via $T_{R}$ and $T'_R$) and provide an upper bound their flip distance formulated \cref{prop:upper} which we first stated in \cref{sec:puttogether}.
	
	\upper*
	
	\begin{proof}
		By \cref{lem:indirect} the flip distance of $T$ and $T'$ is at most $\beta(\lvert \Gamma \rvert-ac(H))$ flips more than the flip distance of $T_R$ and $T'_R$. Every iteration of the direct phase takes at most $2(n-3)$ flips in the set up (\cref{lem:symmdiff}), followed by $\beta+1$ direct flips (\cref{lem:directphase}). There are $ac(H)$ iterations in the direct phase. After that, we obtain the triangulation denoted by $T_{R,ac(H)}$. This triangulation coincides with $T'_R$ on all of $\Lambda(\bar{\Delta}_j\bar{\Delta}'_j)$ for indirect pairs and all $\Lambda(\Delta'_j)\cup \Delta'_j$ for all direct pairs. We conclude that there are again at most $n-3$ edges which may be different between the two triangulations. These can, again by \cref{lem:symmdiff}, be flipped into one another in at most $2(n-3)$ flips. The upper bound on the flip distance follows from adding up all flips from the individual steps and estimating $ac(H)+1\leq n$. The flip distance of $\beta T$ and $\beta T'$ is at most
		\vspace{-1pt}
		
		\begin{equation*}
			2\beta(\lvert \Gamma \rvert-ac(H)) + 2(n-3) (ac(H)+1) + (\beta+1) ac(H) \leq \beta(2\lvert \Gamma \rvert-ac(H))+n(2n-5).
		\end{equation*}
	\end{proof}
		
	\section{Conflict Graph based Lower Bounds on the Flip Distance}\label{sec:lowerbound}
	
	We now want to turn our attention to the study of lower bounds on the length of flip sequences between blown up triangulations $\beta T$ and $\beta T'$. We denote a flip sequence $F$ by $\beta T = T_0, T_1, \ldots ,T_\ell =  \beta T'$. We again consider the pairing $\Gamma$ of triangles in $T$ with triangles of~$T'$ that share a common edge on the spine.
	
	For a pair of triangles $(\Delta,\Delta')\in \Gamma$ of we define $gone(\Delta)$ to be the index of the first triangulation $T_{gone(\Delta)}$ in $F$ for which $T_{gone(\Delta)}$ does not contain any edge of $\Lambda(\Delta)$. Since~$\beta T'$ does not contain an edge of $\Lambda(\Delta)$, the notion of $gone(\Delta)$ is well-defined.
	
	We say that a pair $(\Delta,\Delta')$ is \emph{direct} if there is an index $a\leq gone(\Delta)$ for which $T_a$ contains an edge of $\Lambda(\Delta')$, and \emph{indirect} otherwise.
	
	\begin{lemma}\label{lem:directededge}
		Let $(\Delta_i,\Delta_i')\neq(\Delta_j,\Delta_j')$ both be direct pairs such that the conflict graph $H$ of $T,T'$ contains a directed edge $(\Delta_i,\Delta_i')\to(\Delta_j,\Delta_j')$. Then $gone(\Delta_i)< gone(\Delta_j)$.
	\end{lemma}
	
	\begin{proof}
		Assume for a contradiction that $gone(\Delta_i) > gone(\Delta_j)$. Note that ``=" is not possible since $\Lambda(\Delta_i)$ and $\Lambda(\Delta_j)$ are disjoint and every flip only removes a single edge. Then $T_{gone(\Delta_j)}$ contains an edge $e_i$ of $\Lambda(\Delta_i)$. 
		Further, $T_{gone(\Delta_j)}$ contains an edge $e'_j$ of $\Lambda(\Delta'_j)$, since $(\Delta_j,\Delta_j')$ is a direct pair. However, since $H$ contains a directed edge $(\Delta_i,\Delta_i')\to(\Delta_j,\Delta_j')$, the two edges $e'_j$ and $e_i$ cross. This contradicts the planarity of a triangulation.
	\end{proof}
	
	For a given flip sequence $F$ from $\beta T$ to $\beta T'$, we write $\delta_F$ for the number of direct pairs and~$\bar{\delta}_F$ for the number of indirect pairs. As a next result we bound the number of direct pairs in the size of the largest acyclic subset, $ac(H)$.
	
	\begin{lemma}\label{lem:direct}
		Let $F$ be a flip sequence from $\beta T$ to $\beta T'$, with $\beta>0$, and $H=H(T,T')$ the conflict graph of $T$ and $T'$ for a given linear representation. Then $\delta_F \leq ac(H)$.
	\end{lemma}
	
	\begin{proof}
		Let $(\Delta_1,\Delta'_1), \ldots, (\Delta_{\delta_F},\Delta'_{\delta_F}) \in \Gamma$ be the direct pairs, labeled such that $gone(\Delta_1) < \ldots < gone(\Delta_{\delta_F})$. Then, by \cref{lem:directededge}, every edge of $H$, $(\Delta_i,\Delta'_i)\to(\Delta_j,\Delta'_j)$, fulfills $i<j$. We conclude that $\{(\Delta_1,\Delta'_1), \ldots, (\Delta_{\delta_F},\Delta'_{\delta_F})\}$ induces an acyclic subgraph of $H$ with size $\delta_F$.
	\end{proof}
	
	Now that we have an upper bound on the number of direct pairs and hence a lower bound on the number of indirect pairs in terms of the size of a maximum acyclic set of the conflict graph, we can start deriving lower bounds on the flip distance between $\beta T$ and $\beta T'$. This will be done by showing that the existence of any indirect pair implies the existence of many edges that appear in the flip sequence, but do not lie in $\beta T$ or $\beta T'$. Further, different indirect pairs give rise to disjoint edge sets. For the following lemma, we use the notation~$E(G)$ to refer to the set of edges of a graph $G$.
	
	Further, we will need the notion of an \emph{ear} of a triangulation. An ear of a triangulation of a polygon is a triangle that is bounded by two edges of the boundary of the polygon. It is well known that every triangulation of a polygon with at least $4$ vertices contains at least two ears~\cite{Meisters1975}. We further collect some helpful properties of triangulations in terms of ears. For this purpose we define the \emph{tip} of an ear of a triangulation as the vertex of an ear that is incident to two boundary edges of the polygon.
	
	\begin{observation}\label{obs:ear1}
		Let $v_1$, $v_2$, and $v_3$ be three vertices of a triangulation $\mathscr{T}$ that form an ear with $v_2$ as a tip.
		Removing $v_2$ and its two incident edges from $\mathcal{T}$ gives again a triangulation $\mathcal{T}'$ of a smaller polygon. The edge $v_1v_3$ becomes a boundary edge in $\mathcal{T}'$.
	\end{observation}
	
	\pagebreak
	
	\begin{observation}\label{obs:ear2}
		Let $\mathcal{T}$ be a triangulation of a convex polygon. If we remove the boundary edges of $\mathcal{T}$ from $\mathcal{T}$, then we obtain the following connected components in the resulting graph:
		\begin{itemize}
						\setlength\itemsep{0em}
			\item One big connected component consisting of all interior edges of $\mathcal{T}$ and their incident vertices.
			\item Isolated vertices that were tips of ears of $\mathcal{T}$.
		\end{itemize}
	\end{observation}
	
	With these preliminaries, we are ready to prove the main technical result of this section.
	
	\begin{lemma}\label{lem:crucial2}
		For every indirect pair $p=(\Delta_p,\Delta'_p)$ in a flip sequence $F$, let $T_p=T_{gone{\Delta_p}}$. There exists a subgraph $G_p$ of $T_p$, a set $K_p$ of edges of $T_p$ and an integer $l_p$ such that:
		\begin{itemize}
			\setlength\itemsep{0em}
			\item $K_p$ contains $l_p$ edges of $T_p$ where each edge connects non-consecutive vertices of $V(\Delta_p)$.
			\item $G_p$ contains $\beta-l_p$ vertices of $V^O(\Delta_p)$, no edges from $\beta T\cup \beta T'$, and has at most $3n$ connected components
			\item $K_p\cap E(G_p) = \emptyset$.
		\end{itemize}
		Furthermore, if $p$ and $q=(\Delta_q,\Delta'_q)$ are different indirect pairs of $F$ then:
		\begin{itemize}
			\setlength\itemsep{0em}
			\item $K_p\cap E(G_q) = \emptyset$.
			\item $K_p\cap K_q = \emptyset$.
		\end{itemize}
	\end{lemma}
	
	\begin{proof}
		We repeat the following constructions for every indirect pair $p=(\Delta_p,\Delta'_p)$ to obtain $K_p$ and $G_p$.
		
		Since $p$ is an indirect pair, $T_p \cap \Lambda(\Delta_p) = T_p \cap \Lambda(\Delta') = \emptyset$. Since the only interior edges in $\beta T\cup \beta T'$ that share vertices with $V^O(\Delta_p)$ are the ones in $\Lambda(\Delta_p)$ and~$\Lambda(\Delta'_p)$, no interior edge of $T_p$ that is incident to a vertex in $V^O(\Delta_p)$ is in $\beta T \cup \beta T'$.
		
		We iteratively construct a graph $G_1$ from $T_p$ by removing vertices and edges from $T_p$. First, for every edge $e \in T_p \cap (\beta T \cup \beta T')$ we observe that all points of $V^O(\Delta_p)$ lie on the same side of~$e$. We delete from $T_p$ all the points (and their incident edges) that are contained in the side of $e$ that does not contain $V^O(\Delta_p)$.
		
		Note the resulting graph $G_1$ is a triangulation of a connected convex polygon and contains all vertices of $V^O(\Delta_p)$. Further, all edges of $G_1 \cap (\beta T \cup \beta T')$ lie on the boundary of $G_1$.
		
		Our ultimate goal is for the graph $G_p$ to not contain edges from $\beta T\cup \beta T'$. All remaining such edges lie on the boundary of $G_1$. We will get rid of these edges by removing all the edges on the boundary. But before we can do that we need to make some preparations. Since ears have two edges from the boundary of a triangulation, removing the convex hull  will create an isolated vertex from every tip of an ear. In the following procedure, we deal with this matter to keep the number of connected components low.
		
		We look at all the ears of $G_1$ in which all three vertices of the ear are part of $V(\Delta)$ for some $\Delta$ with $(\Delta,\Delta')\in \Gamma$. We will remove such ears one by one from the triangulation. In the process, we potentially introduce new ears. We repeat the following procedure exhaustively, until there are no longer ears with all three vertices in $V(\Delta)$ for some $\Delta$. For every ear $\Delta_{ear}$ on vertices $v_1$, $v_2$, and $v_3$ in $V(\Delta)$ with tip $v_2$ we do the following:
		
		\begin{itemize}
						\setlength\itemsep{0em}
			\item Remove $v_2$ and its incident edges from $G_1$.
			\item If $\Delta = \Delta_p$, then we add the edge $v_1v_3$ (which is now on the boundary of $G_1$) to $K_p$.
		\end{itemize}
		
		We call the resulting graph $G_2$. Note that $G_2$ is again a triangulation of a connected convex polygon by \cref{obs:ear1}.
		
		The cardinality of $K_p$, denoted by $l_p$, is exactly the number of vertices that we removed from~$V^O(\Delta_p)$.
		
		Further, all edges that connect two vertices of $V(\Delta)$ for any $\Delta$ are either on the boundary of~$G_2$ or are no longer contained in $G_2$.
		
		To observe this assume for a contradiction that there was a remaining interior edge between two vertices of $V(\Delta)$ in the interior of $G_2$. Splitting the triangulation along this edge gives two smaller triangulations. One of these two triangulations exists entirely on points of $V(\Delta)$. We know that this triangulation must have at least two ears, one of which has all its vertices in $V^O(\Delta)$ such that the vertices appear consecutively along the boundary of~$G_2$. This particular ear was already present in~$G_2$, a contradiction.
		
		We conclude that all ears of $G_2$ now contain at most two vertices of some $V(\Delta)$ for any $\Delta$. In particular, every ear has a third vertex that is not in $V(\Delta)$. This allows us to give a very rough estimate on the number of ears in $G_2$ that only depends on $n$, but not on $\beta$.
		
		First, we bound the number of triangles that have vertices in $V(\Delta)$ for some $\Delta$ with $(\Delta,\Delta')\in~\Gamma$. Such ears can only appear when the non-boundary edge of the triangulation connects two vertices where not both of them belong to the same set $V(\Delta)$. This happens once on every side of the points $V(\Delta)$ and at most $2\lvert \Gamma \rvert < 2n$ times in total.
		
		Next, we bound the number of ears that do not contain vertices of any $V^O(\Delta)$. There are at most $n$ vertices that are not in any $V^O(\Delta)$. With these $n$ vertices we can build at most $n/2$ ears. We sum up and conclude that there are less than $3n$ ears in $G_2$.
		
		Finally, we are ready to remove the boundary of $G_2$. We set $G_p$ to be the graph that results from removing the boundary from $G_2$. $G_p$ has at most $3n$ connected components. One connected component contains all the interior edges of $G_2$ and the others are isolated vertices that were formerly tips of ears.
		
		Further, $G_p$ does not contain any edges from $\beta T\cup \beta T'$ since they were all removed, when removing the edges on the boundary. $G_p$ contains all vertices from $V^O(\Delta_p)$ that have not been removed while removing ears. For every vertex of $V^O(\Delta_p)$ that was removed while removing an ear, an edge has been added to $K_p$. Therefore, if $K_p$ contains $l_p$ edges, then $G_p$ contains $\beta-l_p$ vertices from~$V^O(\Delta_p)$.
		
		$E(G_p)$ and $K_q$ for any $q$ are disjoint, since any edge that has two vertices in common with $V^O(\Delta_q)$ has either been removed when removing ears, or when removing the boundary of~$G_2$, or was never in $T_p$ in the first place. Further, the sets $K_p$ and $K_q$ are disjoint, since their endpoints of their edges are different.
	\end{proof}
	
	Next, we use \Cref{lem:crucial2} to prove a lower bound on the length of a flip sequence~$F$.
	
	\begin{lemma}\label{lem:lower}
		Any flip sequence $F$ from $\beta T$ to $\beta T'$ has length at least $\beta (\lvert\Gamma\rvert + \bar{\delta}_F) - 3n\bar{\delta}_F$ = $\beta (2\lvert\Gamma\rvert-\delta_F) - 3n\bar{\delta}_F$.
	\end{lemma}
	
	\begin{proof}
		For every indirect pair $(\Delta_i,\Delta'_i)$ with $i$ ranging from $1$ to $\bar{\delta}_F$ let $G_i$ and $K_i$ be the respective subgraph and set of edges obtained by~\cref{lem:crucial2} and $l_i$ the cardinality of $K_i$. Further, let $L$ denote $\sum_{i=1}^{\bar{\delta}_F} l_i$. As $K_i$ and $K_j$ are disjoint for $i\neq j$, we get that $L = \lvert \cup_{i=1}^{\bar{\delta}_F} K_i \rvert$.
		
		Let $U$ be the union of all $G_i$. Then $U$ contains at least $\beta\bar{\delta}_F-L$ vertices, since it contains~$\beta-l_i$ vertices of $V^O(\Delta_i)$ and these sets of vertices are disjoint for different $i$. Further, $U$ has at most $3n\bar{\delta}_F$ connected components. We conclude that $U$ contains at least $\beta\bar{\delta}_F-(L+3n\bar{\delta}_F)$ distinct edges that are neither part of $\beta T$ nor $\beta T'$. By also counting the edges in $K_i$ we get a total of $\beta\bar{\delta}_F-3n\bar{\delta}_F$ edges that have been introduced during the flip sequence~$F$, but are neither part of $\beta T$ nor $\beta T'$. We can add the cardinalities of the two sets since $E(G_i)$ and $K_j$ are disjoint for all choices of $i,j\in\{1,\ldots,\bar{\delta}_F\}$.
		
		Now $\beta T\setminus \beta T'$ contains at least $\beta\Gamma$ edges that all need to get flipped during $F$. Further, there are $\beta\bar{\delta}_F-3n\bar{\delta}_F$ edges that are neither in $\beta T$ nor in $\beta T'$, but appear as edges in $F$. All these edges need to be removed by some flip in $F$. This gives a total of
		\begin{equation*}
			\beta\lvert\Gamma\rvert + \beta\bar{\delta}_F-3n\bar{\delta}_F = \beta (\lvert\Gamma\rvert + \bar{\delta}_F) - 3n\bar{\delta}_F = \beta(2\lvert\Gamma\rvert-\delta_F) - 3n\bar{\delta}_F
		\end{equation*}
		flips in the flip sequence $F$. For the last equality, we use that $\lvert\Gamma\rvert = \delta_F + \bar{\delta}_F$.
	\end{proof}
	
	We obtain our final lower bound that is independent of the flip sequence $F$, by using \cref{lem:direct} to bound the number of direct pairs.
	
	\lower*
	
	\begin{proof}
		We modify the lower bound from \Cref{lem:lower}. We use \cref{lem:direct} to upper bound $\delta_F$ from above with $ac(H)$. Further, $\bar{\delta}_F \leq \lvert\Gamma\rvert \leq n$.
	\end{proof}
	
	\subparagraph{Acknowledgements and Funding.} Research on this project is funded by the Austrian Science Fund (FWF) 10.55776/DOC183. Helpful discussions on this topic took place during workshops and events associated to the PAGCAP project (Austrian Science Fund FWF, grant~I~5788).
	
	I want to thank Anna Hofer and Birgit Vogtenhuber for helpful comments regarding the write-up, Cesar Ceballos for helping with finding related work about lattices and polytopes, Alex Black and Raphael Steiner for pointing out the connection between my result and the computation of shortest paths in simple polytopes, and Oswin Aichholzer for drawing my attention towards this problem and for previous joint research on reconfiguration problems. Further, I want to thank the anonymous reviewers for helpful remarks regarding the illustration of the paper and for pointing out additional connections to related topics.

	\bibliographystyle{alpha} 
	\bibliography{references}

\end{document}